\documentclass[12pt]{article}
\usepackage{graphicx}
\usepackage{cmap}
\usepackage{xcolor}
\usepackage{cite}
\usepackage{array}
\usepackage{float} 
\usepackage{amsmath,amssymb}
\usepackage{geometry}

\newtheorem{definition}{Definition}

\newtheorem{remark}{Remark}
\newtheorem{theorem}{Theorem}
\newtheorem{corollary}{Corollary}
\newtheorem{example}{Example}
\newtheorem{lemma}{Lemma}
\newtheorem{proof}{Proof}

\newcommand{\R}{\mathbb{R}}

\newcommand{\divergence}{\operatorname{div}}
\newcommand{\tr}{\operatorname{tr}}

\begin{document}  %%%!!!

\title{\MakeUppercase{Weighted Phase Volume Method in Stability Analysis: Integral Criteria and Ellipsoidal Reachable Sets}}%

\author{Igor B.~Furtat}

%\date{2026}

\maketitle

\begin{center}
\textit{Institute for Problems in Mechanical Engineering, Russian Academy of Sciences, St. Petersburg, Russia}
\\
cainenash@mail.ru
\end{center}

\begin{abstract}
A method for analysing the stability of dynamical systems is proposed, based on the introduction of a weighted phase volume and time rescaling by a positive function. 
The advantage of the method is the ability to set the contraction properties of the phase volume by choosing the weighting function and the scaling factor, while preserving the topology of the phase portrait. 
Integral dissipativity conditions are derived, leading to new definitions of integral stability, asymptotic stability, and exponential stability. 
For quadratic weighting functions, covering and inner ellipsoids are constructed, providing geometric estimates of reachable sets. The connection between the proposed approach and classical Lyapunov stability is established. 
The efficiency of the method is demonstrated through numerical examples.
\end{abstract}

%%%%%%%%%%%%%%%%%%%%%%%%%%%%%%%%%%%%%%%%%%%%%%%%%%%%%%

\section{Introduction} 

The stability of dynamical systems is still one of the important problems in differential equation theory and mathematical physics.  
The classical Lyapunov method \cite{Lyapunov50} and its developments in \cite{Chetaev55,Liotov62,Malkin66,Zubov84} provide effective tools for analysing equilibrium stability. 
However, constructing Lyapunov functions for complex, nonlinear, nonautonomous or discontinuous systems remains a significant challenge.

Alternative approaches based on the geometric properties of the vector field go back to the classical Liouville theorem~\cite{Liouville,Arnold73,Arnold} and the Reynolds transport theorem~\cite{Reynolds}. 
They describe the evolution of phase volume and establish a link between stability and the sign of divergence in vector fields. 
Pioneering works \cite{Zaremba54,Fronteau65,Brauchli68} laid the foundations of divergence stability analysis. In \cite{Shestakov78}, index and divergence criteria for the stability of a singular point were obtained. 
In a series of works by V.P. Zhukov~\cite{Jukov78,Jukov79,Jukov90,Jukov99}, the source and sink method is developed, enabling the formulation of necessary and sufficient conditions for the instability and asymptotic stability of nonlinear autonomous systems. 
However, these results are often limited by the dimension of the phase space, or they require specific assumptions about the structure of the vector field.

A significant step is taken by A. Rantzer in \cite{Rantzer00,Rantzer01}, where the concept of dual Lyapunov functions (density method) is introduced. 
This stability concept is proposed for almost all initial conditions \cite{Rantzer00,Rantzer01,Monzon03}. 
Linear matrix inequalities are also introduced to verify these conditions.

Further development of \cite{Jukov78,Jukov79,Jukov90,Jukov99,Rantzer00,Rantzer01,Monzon03} is presented in \cite{Furtat20a,Furtat20,Furtat21,Furtat22,Furtat23,Furtat26}. 
These developments include divergence stability conditions and methods for analysing nonautonomous and perturbed systems, as well as the theory of density systems. 
In contrast to the classical Lyapunov method \cite{Lyapunov50,Chetaev55,Liotov62,Malkin66,Zubov84}, the approaches  \cite{Jukov78,Jukov79,Jukov90,Jukov99,Rantzer00,Rantzer01,Monzon03,Furtat20a,Furtat20,Furtat21,Furtat22,Furtat23,Furtat26} do not provide direct geometric information about reachable sets.

Several studies have focused on the analysis of dissipativity and stability of nonconservative systems. 
In \cite{Willems72}, the concept of dissipative dynamical systems with quadratic supply rates is introduced. 
In \cite{Bikdash00}, an energy-based approach to constructing Lyapunov functions for physical systems is proposed. 
In \cite{Yuan14}, a dynamical equivalence between a Lyapunov function and a potential function is established. 
In \cite{Kounadis03}, loss of stability of nonconservative systems in regions of divergence instability is investigated. 
In \cite{Zhu21}, a new dissipativity criterion based on the notion of dissipative power was proposed, surpassing the sensitivity of the classical divergence criterion. 
In \cite{FE_Stability25,Taghvaei20,Feiling21,Efimov22,Zamani22}, modern numerical and analytical methods for stability analysis are presented, including finite-element analysis of nonconservative systems, the use of Sinkhorn divergence to assess filter robustness, the divergence theorem in source-seeking problems, convergence conditions for generalised Persidskii systems, and criteria for incremental stability. 
In \cite{Pakshin2024_1,Pakshin2023_1}, iterative learning control methods for stochastic multi-agent systems are developed.
However, most these approaches are either restricted to special classes of systems or do not provide constructive geometric estimates of the evolution of reachable sets.

This paper introduces the concept of a weighted phase volume, based on using an arbitrary weighting function $W(x,t)$ and a matrix-valued scaling function $\rho(x)$, performing a time change $d\tau = \mu(x)dt$ with the scalar factor $\mu=\langle F,\rho F\rangle/\|F\|^2$. 
While the classical Liouville theorem~\cite{Liouville,Arnold73,Arnold} describes the evolution of the Euclidean volume and the Reynolds transport theorem~\cite{Reynolds} describes the evolution of the integral of an arbitrary function along the field $F$, the proposed approach studies the evolution of the integrals $\int_{\Omega_t} W\,dx$ and $\int_{\Omega_t} 1/W\,dx$ in scaled time. 
This allows, while preserving the topology of the phase portrait (unlike methods that alter the vector field itself), to purposefully set the dissipativity properties and obtain stability conditions that explicitly depend on the choice of $\rho(x)$ and $W(x,t)$.

The main results of the paper and their differences from known works are formulated as follows:

%\begin{itemize}
(i) In Theorem~\ref{Th:dissip_cond}, new integral identities are obtained. 
Unlike the classical Liouville theorem~\cite{Liouville} and the Reynolds transport theorem~\cite{Reynolds}, these identities depend on arbitrary weighting functions and time scaling, allowing us to expand the class of systems under study.  
Compared to the divergence conditions of Zhukov~\cite{Jukov78,Jukov79}, which operate with pointwise divergence estimates, the proposed identities are of integral nature and include the derivative of the weighting function along trajectories.

(ii) Paired $\rho$-dissipativity conditions for $W$ and $1/W$ are formulated (Section~\ref{Sec:Int_relat}). Unlike the classical concept of dissipativity according to Willems~\cite{Willems72}, which uses a fixed storage function, and the dissipative power criterion in \cite{Zhu21}, the proposed conditions are parameterised by two free functions $\rho$ and $W$, significantly expanding the class of analysable systems.

(iii) New definitions of integral stability, asymptotic stability, and exponential stability with respect to the pair $(\mu,W)$ are introduced (Definitions~\ref{def:integral_stability}--\ref{def:integral_exp_stability}). 
Unlike the density method of Rantzer~\cite{Rantzer00,Rantzer01} and developed in~\cite{Furtat23}, which considers density evolution, the proposed definitions explicitly include the scaling factor $\mu(x)$ and allow two-sided estimates via $W$ and $1/W$, which is absent in previous works on integral stability~\cite{Monzon03}.

(iv) A method for constructing approximating ellipsoids is developed. 
Two new dissipativity conditions are obtained.
From the first dissipativity condition, an evolution equation for the covering ellipsoid is derived (Theorem~\ref{Th:Stability}), which is guaranteed to contain the domain $\Omega_t$. 
From the second condition, an equation for the inner ellipsoid is obtained (Theorem~\ref{th:inner_ellipsoid_general}), which is free of points from $\Omega_t$. Unlike known ellipsoidal approximation methods (e.g., those based on linear matrix inequalities and reachable sets), the proposed approach does not require solving optimization problems at each step. 
It reduces to integrating differential equations for the centre and shape matrix. 
Using both covering and inner ellipsoids together provides a complete geometric picture of the evolution in the form of an ellipsoidal annulus, which is absent in previous works on divergence stability analysis \cite{Furtat20a,Furtat21,Furtat22}.

(v) A connection is established between the proposed integral stability and classical Lyapunov stability (Theorem~\ref{Th:Lyap}). 
It is demonstrated that the presence of a contracting family of covering ellipsoids implies asymptotic (or exponential) Lyapunov stability. 
This generalises the results of~\cite{Furtat20,Furtat21}, in which the connection with Lyapunov stability is only established for specific system classes and without geometric estimates.

(vi) The efficiency of the method is demonstrated on examples of second-order nonlinear systems, where the divergence method with $\rho=E$ ($E$ is the identity matrix) fails to establish stability. 
However, a suitable choice of the scaling function $\rho(x)$ ensures the fulfilment of generalised dissipativity conditions. 
Unlike numerical approaches~\cite{Taghvaei20,Feiling21,FE_Stability25}, the proposed method provides analytical stability guarantees and constructive geometric estimates.
%\end{itemize}

The paper is organised as follows. 
Section~\ref{Probl_St} states the problem, introduces the original and scaled systems, and defines integral quantities and their analogs for the scaled system. 
Section~\ref{sec:prelim} contains the necessary preliminaries on derivatives along the system and the classical Liouville theorem. 
Section~\ref{sec:scaled_properties} establishes the main properties of the scaled system and its connection with the original one. 
Section~\ref{Sec:Int_relat} derives integral identities for the evolution of weighted volume, formulates generalised $\rho$-dissipativity conditions, and provides their physical interpretation. 
Section~\ref{sec:integral_stability} introduces new definitions of integral stability with respect to the pair $(\mu,W)$ and conducts a local analysis near an equilibrium. 
Sections~\ref{sec:outer_ellipsoid} and~\ref{sec:inner_ellipsoid} are devoted to constructing covering and inner ellipsoids. 
The corresponding evolution equations for the shape matrix are derived, and an illustrative example is given. 
Section~\ref{sec:lyapunov_connection} establishes the connection between the proposed integral stability and classical Lyapunov stability. 
The concluding Section~\ref{sec:conclusion} summarizes the main results of the paper.

The following \textit{notations} are used in this paper:
$\mathbb R^{n}$ is an $n$-dimensional Euclidean space with norm $\|\cdot\|$;
$\mathbb R_+=[0,\infty)$; 
the notation $\phi(x) = O(\|x\|^k)$ as $x \to 0$ means that $\limsup_{x \to 0} \frac{\|\phi(x)\|}{\|x\|^k} < \infty$;
the notation $\phi(x) = o(\|x\|^k)$ as $x \to 0$ means that $\lim_{x \to 0} \frac{\|\phi(x)\|}{\|x\|^k} = 0$;
the symbol ``$^\top$'' denotes the transpose operation;
the notation $C^l$ denotes the class of functions that have continuous $l$th derivatives; 
$E$ is the identity matrix;
$\operatorname{Tr}A = \sum_{i=1^n}a_{ij}$ is a trace of the matrix $A$;
``$\wedge$'' stands for the wedge product (exterior product), which is an antisymmetric multiplication operation on differential forms or vectors;
$\langle \cdot,\, \cdot \rangle$ is an inner  product.

%%%%%%%%%%%%%%%%%%%%%%%%%%%%%%%%%%%%%%%%%%%%%%%%%%%%%%%%

\section{Problem Statement}
\label{Probl_St}

Consider an autonomous dynamical (original) system of the form
\begin{equation}
\label{eq:original}
\dot{x}(t) = F(x(t)),
\end{equation}
where $t \in \mathbb R_+$, 
$x \in \mathbb{R}^n$ is the state vector,
$F : \mathbb{R}^n \to \mathbb{R}^n$ is a smooth vector field ($F \in C^1(\mathbb{R}^n)$). 
If a different order of smoothness is required, it will be indicated separately when deriving the corresponding result. 
Assume that the system \eqref{eq:original} generates a global phase flow $\varphi_t : \mathbb{R}^n \to \mathbb{R}^n$, i.e., for any $x_0 \in \mathbb{R}^n$, the Cauchy problem with initial condition $x(0)=x_0$ has a unique solution $x(t)=\varphi_t(x_0)$, defined for all $t\in\mathbb{R}_+$.

Introduce an auxiliary (scaled) system
\begin{equation}
\label{eq:scaled}
\frac{dx}{d\tau} = G(x(\tau)) := \rho(x(\tau)) F(x(\tau)),
\end{equation}
where $\tau \ge 0$ is time in the scaled system, 
$\rho: \mathbb{R}^n \to \mathbb{R}^{n \times n}$ is a smooth matrix-valued function satisfying the following collinearity and positivity conditions for every $x$:
\begin{equation}
\label{eq:collinearity_conditions}
\begin{array}{lll}
F(x)\wedge\bigl(\rho(x)F(x)\bigr)=0,
\\
\bigl\langle F(x),\,\rho(x)F(x)\bigr\rangle >0\;\; \text{whenever } F(x)\neq 0.
\end{array}
\end{equation}
These conditions guarantee that $G(x)$ is a positive scalar multiple of $F(x)$ at each point, i.e.\ there exists a positive scalar function
\begin{equation}\label{eq:mu_def}
\mu(x) = \frac{\langle F(x),\rho(x)F(x)\rangle}{\|F(x)\|^2},~~ \mu(x)>0\ \text{for}\ F(x)\neq 0,
\end{equation}
such that $G(x) = \mu(x) F(x)$. The function $\mu(x)$ plays the role of the time‑scaling factor.

Consider an arbitrary Lebesgue-measurable domain $\Omega_0 \subset \mathbb{R}^n$ with finite measure. 
Under the flow $\varphi_t$ of the system \eqref{eq:original}, the domain $\Omega_0$ evolves
$\Omega_t = \varphi_t(\Omega_0) = \{x \in \mathbb{R}^n : x = \varphi_t(\xi), \ \xi \in \Omega_0\}$.
Similarly, for the system \eqref{eq:scaled}, we denote $\Omega_\tau^G = \varphi_\tau^G(\Omega_0)$.

Let $W: \mathbb{R}^n \times R_+ \to \mathbb{R}$ be an arbitrary function of class $C^1(\mathbb{R}^n \times R_+)$ (additional conditions required for various applications will be specified later).
Introduce integral expressions for the original system \eqref{eq:original} in the forms
\begin{equation*}
\label{eq:integrals_original}
I(t) = \int_{\Omega_t} W(x,t)\,dx, 
~~ 
I_{-1}(t) = \int_{\Omega_t} \frac{1}{W(x,t)}\,dx. 
\end{equation*}
For the system \eqref{eq:scaled}, we similarly define
\begin{equation}
\label{eq:integrals_scaled}
I^G(\tau) = \int_{\Omega_\tau^G} W(x,\tau)\,dx, 
~~
 I_{-1}^G(\tau) = \int_{\Omega_\tau^G} \frac{1}{W(x,\tau)}\,dx. 
\end{equation}
Here $dx$ denotes the volume element in $\mathbb{R}^n$. 
In the second integral expressions, it is assumed that $W(x,t) \neq 0$ for all $x \in \Omega_t$ and $t \in \mathbb R_+$, as well as $W(x,\tau) \neq 0$ for all $x \in \Omega_{\tau}^G$ and $\tau \in \mathbb R_+$.

The main goal of the paper is to derive expressions for the derivatives $dI/dt$ and $dI_{-1}/dt$, as well as $dI^G/d \tau$ and $dI_{-1}^G/d \tau$. 
From these, conditions will be obtained that connect the behaviour of the original system \eqref{eq:original} with that of the scaled system \eqref{eq:scaled}. 
On this basis, stability criteria will be formulated.

%%%%%%%%%%%%%%%%%%%%%%%%%%%%%%%%%%%%%%%%%%%%%%%%%%%%%%%%

\section{Preliminaries}
\label{sec:prelim}

%\subsection{Derivative along the system}

Define the derivative of the function $W(x,t)$ along the trajectories of the system \eqref{eq:scaled} as the Lie derivative of $W$ along the vector field $G$:
\begin{equation}
\label{eq:lie_derivative_scaled}
\begin{array}{lll}
\dot{W}_G(x,t) = \frac{\partial W}{\partial t}(x,t) + \nabla W(x,t)^{\top} G(x)  
= \frac{\partial W}{\partial t} + \nabla W(x,t)^{\top} \rho(x) F(x). 
\end{array}
\end{equation}
The notation $\nabla W(\varphi_t(\xi),t)$ may appear in the paper, meaning $\nabla W(\varphi_t(\xi),t) = \nabla W(x,t)|_{x=\varphi_t(\xi)}=col\left\{\frac{\partial W}{\partial x_1}(x,t),..., \frac{\partial W}{\partial x_n}(x,t)\right\}$.

%\subsection{Liouville's Theorem}

Let $J_t(\xi) = \partial \varphi_t(\xi)/\partial \xi$ be the Jacobian matrix of the flow. It is known that $J_t(\xi)$ satisfies the variational equation
\begin{equation}\label{eq:variational}
\frac{d}{dt} J_t(\xi) = A(\varphi_t(\xi)) J_t(\xi),\quad J_0(\xi) = E,
\end{equation}
where $A(x) = \partial F(x)/\partial x$ is the Jacobian matrix of the vector field.

\begin{theorem}[Liouville's Theorem \cite{Liouville,Arnold,Hartman64,Arnold73}] 
\label{Th:Liouville}
The following relation holds:
\begin{equation*}
\label{eq:liouville}
\frac{d}{dt} \det J_t(\xi) = \det J_t(\xi) \operatorname{tr} A(\varphi_t(\xi)).
\end{equation*}
\end{theorem}

Since $\operatorname{Tr} A(x) = \operatorname{div} F(x)$, from \eqref{eq:liouville} it follows that
\begin{equation}\label{eq:liouville_div}
\frac{d}{dt} \det J_t(\xi) = \det J_t(\xi) \operatorname{div} F(\varphi_t(\xi)).
\end{equation}
Here and below, $\operatorname{div} F(\varphi_t(\xi))=\operatorname{div} F(x)|_{x=\varphi_t(\xi)}$.

\begin{remark} 
For flows generated by system \eqref{eq:original}, the Jacobian determinant is always positive. 
Indeed, the equation \eqref{eq:liouville_div} is a first-order linear homogeneous equation with respect to $\det J_t$, whose solution has the form
\begin{equation*}
\label{eq:det_solution}
\det J_t(\xi) = \exp\left( \int_0^{t} \operatorname{div} F(\varphi_s(\xi)) \, ds \right) > 0 \quad \forall t.
\end{equation*}
Therefore, in the sequel, the absolute value sign will be omitted when using the Jacobian determinant.
\end{remark}

%%%%%%%%%%%%%%%%%%%%%%%%%%%%%%%%%%%%%%%%%%%

\section{Properties of the Scaled System}
\label{sec:scaled_properties}

Formulate a theorem summarizing the main properties of the connection between the original system \eqref{eq:original} and the scaled system \eqref{eq:scaled}.

\begin{theorem}
\label{Th:main_prop_systems}
Let $\rho:\mathbb{R}^n\to \mathbb{R}^{n\times n}$ be a smooth matrix-valued function satisfying the collinearity and positivity conditions~\eqref{eq:collinearity_conditions}, and let $\mu(x)$ be defined by \eqref{eq:mu_def}. Define the scaled system by \eqref{eq:scaled}. Then the original system \eqref{eq:original} and the scaled system \eqref{eq:scaled} possess the following properties:
%\begin{enumerate}
%\item

1)  The sets of equilibria coincide, i.e.\ 
$F(x^*)=0 \;\Leftrightarrow\; G(x^*)=0$.

2) Let $x=0$ be an equilibrium, $F(0)=0$, and let
$F(x)=Ax+o(\|x\|)$ with $A=\frac{\partial F}{\partial x}(0)$.
Then the linearisation matrix of \eqref{eq:scaled} at the origin is
$A_G=\rho(0)A$, and there exists a constant $\mu_0>0$ such that
$A_G=\mu_0 A$.
Consequently, the eigenvalues are related by
$\lambda_i(A_G)=\mu_0\lambda_i(A)$, the signs of their real parts are
preserved, and the type of stability in the first approximation is the
same for both systems (the convergence rate near the equilibrium is
multiplied by $\mu_0$).

3) The trajectories of the two systems coincide as geometric curves
(up to reparametrisation).  For any point with $F(x)\neq0$ the time
parameters $t$ and $\tau$ are related by
\begin{equation}\label{eq:time_scaling}
\frac{d\tau}{dt}= \mu(x),
\qquad
\tau(t)=\int_0^t \mu(x(s))\,ds,
\end{equation}
where $\mu(x)$ is defined in \eqref{eq:mu_def}.
As a consequence, limit sets, cycles, separatrices and the topology of
the phase portrait are fully preserved.

4) The evolution of phase volumes is governed by
\small
\begin{align}
\frac{d}{dt}\operatorname{vol}(\Omega_t) &=
\int_{\Omega_t}\operatorname{div} F(x)\,dx, \label{eq:liouville_original}\\
\frac{d}{d\tau}\operatorname{vol}(\Omega_\tau^G) &=
\int_{\Omega_\tau^G}\Bigl[(\operatorname{div}\rho(x))\!\cdot\! F(x)
+ \operatorname{Tr}\!\bigl(\rho(x)\nabla F(x)\bigr)\Bigr]dx,
\label{eq:liouville_scaled}
\end{align}
\normalsize
where $(\operatorname{div}\rho)_j = \sum_i \partial_i\rho_{ij}=\sum_{i=1}^n \frac{\partial \rho_{ij}}{\partial x_i}(x)$, $ j=1,\dots,n$ and
$\nabla F$ is the Jacobian matrix of $F$.
The volume change rates are different unless
$(\operatorname{div}\rho)\!\cdot\! F + \operatorname{Tr}(\rho\nabla F)
= \operatorname{div} F$ pointwise.

5) If $\operatorname{div} F=0$ (the original system is conservative), then
$\operatorname{div}(\rho F)=(\operatorname{div}\rho)\!\cdot\! F
+ \operatorname{Tr}(\rho\nabla F)$.  The scaled system is conservative
if and only if this expression vanishes identically.
In the particular case where $\rho$ is a constant scalar matrix
$\rho\equiv c E$ ($c>0$), one has
$\operatorname{div}(\rho F)=c\operatorname{div} F$, and conservativity
of the scaled system is equivalent to conservativity of the original one.

6) The sign of $\operatorname{div}(\rho F)=(\operatorname{div}\rho)\!\cdot\! F
+ \operatorname{Tr}(\rho\nabla F)$ can differ from the sign of
$\operatorname{div} F$.
Even if $\operatorname{div} F<0$ (the original system is dissipative),
the term $(\operatorname{div}\rho)\!\cdot\! F$ may make the scaled system
non-dissipative or even volume-expanding.
Conversely, when $\operatorname{div} F>0$, a suitable choice of $\rho$
satisfying \eqref{eq:collinearity_conditions} can achieve
$\operatorname{div}(\rho F)<0$, i.e.\ render the scaled system
dissipative while preserving all trajectories.
%\end{enumerate}
\end{theorem}

\begin{proof}
We prove each statement using only the conditions \eqref{eq:collinearity_conditions} and the definition of $\mu(x)$ in \eqref{eq:mu_def}.
%\begin{enumerate}

1) If $F(x^*)=0$, then $G(x^*)=\rho(x^*)=0$. Conversely, if
$G(x^*)=0$, then $\langle F(x^*),G(x^*)\rangle =0$.
If $F(x^*)\neq0$, condition \eqref{eq:collinearity_conditions} would give
$\langle F(x^*),\rho(x^*)F(x^*)\rangle >0$, contradicting
$\langle F(x^*),0\rangle=0$. Hence $F(x^*)=0$.

2) From $F(x)=Ax+o(\|x\|)$ we obtain
$G(x)=\rho(x)F(x)=\rho(0)Ax+o(\|x\|)$, because the term
$(\rho(x)-\rho(0))F(x)$ is $O(\|x\|^2)$. Thus
$A_G=\frac{\partial G}{\partial x}(0)=\rho(0)A$.
For any vector $v\neq 0$, take a sequence $x_k=\varepsilon_k v$ with
$\varepsilon_k\to0$, $\varepsilon_k>0$. By
\eqref{eq:collinearity_conditions} and smoothness,
$\rho(\varepsilon_k v)F(\varepsilon_k v)
=\mu(\varepsilon_k v)F(\varepsilon_k v)$. Dividing by $\varepsilon_k$ and
letting $\varepsilon_k\to0$ yields
$\rho(0)Av = \mu_0 Av$ with $\mu_0=\lim_{x\to0}\mu(x)>0$.
Therefore $\rho(0)A = \mu_0 A$, so $A_G=\mu_0 A$, and the eigenvalue
relation follows. Since $\mu_0>0$, the stability type in the first
approximation is unchanged.

3) The conditions \eqref{eq:collinearity_conditions} mean that
$G(x)=\mu(x)F(x)$ with $\mu(x)>0$. Hence the vector fields are
positively collinear, so their integral curves coincide as sets.
The time change \eqref{eq:time_scaling} is obtained by equating the two
parametrisations of the same curve:
$dx/dt = F(x)$, $dx/d\tau = \mu(x)F(x) = \frac{d\tau}{dt}F(x)$,
from which $d\tau/dt = \mu(x)$.

4) Equation \eqref{eq:liouville_original} is the standard
Liouville theorem. For the scaled system, using $G=\rho F$,
$\operatorname{div} G = \sum_i\partial_i\sum_j \rho_{ij}F_j
= \sum_{i,j}\bigl[(\partial_i\rho_{ij})F_j + \rho_{ij}\partial_i F_j\bigr] 
= (\operatorname{div}\rho)\!\cdot\! F + \operatorname{Tr}(\rho\,\nabla F)$.
Inserting this into the Liouville formula gives
\eqref{eq:liouville_scaled}.

5) Follows directly from the divergence formula and the fact that
for $\rho=c E$ one has $\operatorname{div}\rho=0$ and
$\operatorname{Tr}(\rho\nabla F)=c\operatorname{div} F$.

6) The expression for $\operatorname{div}(\rho F)$ contains the
additional term $(\operatorname{div}\rho)\!\cdot\! F$, which can
change the sign independently of $\operatorname{div} F$.
The sign manipulation is possible because one can choose
$\rho(x)$ that satisfies \eqref{eq:collinearity_conditions} while
making $(\operatorname{div}\rho)\!\cdot\! F$ sufficiently negative
(or positive) in the required regions of the phase space.
%\end{enumerate}
\end{proof}

%\begin{remark}
Theorem~\ref{Th:main_prop_systems} shows that a matrix-valued
scaling $\rho(x)$ satisfying the collinearity and positivity
conditions~\eqref{eq:collinearity_conditions} preserves the
qualitative structure of the phase portrait, exactly as a scalar
positive function does. The matrix formulation is often more
convenient in synthesis and robustness problems because it allows
one to impose structural constraints (sparsity, conservation laws,
passivity) directly on $\rho$, and the resulting conditions are
linear in $\rho$ at each point, which is advantageous for
optimisation-based design.
%\end{remark}

%%%%%%%%%%%%%%%%%%%%%%%%%%%%%%%%%%%%%%%%%%%%%%%%%%%%%%%%%%%%%%

\section{Integral Relations for the Scaled and Original Systems}
\label{Sec:Int_relat}

\subsection{Integral Identity for the Scaled System}

\begin{theorem}
\label{Th:dissip_cond}
Let $W \in C^1(\mathbb{R}^n \times \mathbb{R}_+)$, $F \in C^1(\mathbb{R}^n)$, $\rho \in C^1(\mathbb{R}^n, \mathbb{R}^{n\times n})$, and $\Omega_0 \subset \mathbb{R}^n$ be a measurable domain with finite measure. 
Then for the functions $I^G(\tau)$ and $I_{-1}^G(\tau)$ defined by \eqref{eq:integrals_scaled}, the following equalities hold
\small
\begin{equation}
\label{eq:integral_identity_scaled_1}
\begin{array}{lll}
\frac{d}{d\tau} I^G(\tau) = \int_{\Omega_\tau^G} \left[ \dot{W}_G(x,\tau) + W(x,\tau) \divergence G(x) \right] dx,
\\
%\label{eq:integral_identity_scaled_2}
\frac{d}{d\tau} I_{-1}^G(\tau) = \int_{\Omega_\tau^G} \left[ -\frac{\dot{W}_G(x,\tau)}{W^2(x,\tau)} + \frac{1}{W(x,\tau)} \divergence G(x) \right] dx, 
\end{array}
\end{equation}
\normalsize
where $\dot{W}_G$ is defined by \eqref{eq:lie_derivative_scaled}.
\end{theorem}

%%%%%%%%%%%%%%%%%%%%%%%%%%%%

\begin{proof}
Perform a change of variables, going back to the initial domain $\Omega_0$ using the flow $\varphi_\tau^G$ of the system \eqref{eq:scaled}. 
Let $\xi \in \Omega_0$ and $x = \varphi_\tau^G(\xi)$. 
Denote $J_\tau^G(\xi) = \det \frac{\partial \varphi_\tau^G}{\partial \xi}(\xi)$. 
Rewrite \eqref{eq:integrals_scaled} as follows
\begin{equation}
\label{eq:proof_1}
\begin{array}{lll}
I^G(\tau) = \int_{\Omega_0} W(\varphi_\tau^G(\xi), \tau) J_\tau^G(\xi) \, d\xi,
\\
I_{-1}^G(\tau) = \int_{\Omega_0} \frac{1}{W(\varphi_\tau^G(\xi), \tau)} J_\tau^G(\xi) \, d\xi. 
\end{array}
\end{equation}

According to Liouville's theorem (see Theorem \ref{Th:Liouville}), applied to the system \eqref{eq:scaled}, the Jacobian determinant satisfies the equation
$\frac{d}{d\tau} J_\tau^G(\xi) = J_\tau^G(\xi) \operatorname{div} G(\varphi_\tau^G(\xi))$.
Since the integration domain $\Omega_0$ does not depend on $\tau$, taking into account \eqref{eq:lie_derivative_scaled}, we differentiate \eqref{eq:proof_1} with respect to $\tau$:
\begin{equation*}
\label{eq:proof_before_return}
\begin{array}{lll}
\frac{d}{d\tau} I^G(\tau) & =  \int_{\Omega_0} J_\tau^G(\xi) \Big[ \frac{\partial W}{\partial \tau}(\varphi_\tau^G(\xi), \tau)  
\\
& + \nabla W(\varphi_\tau^G(\xi), \tau)^{\top} G(\varphi_\tau^G(\xi)) 
\\
& + W(\varphi_\tau^G(\xi), \tau) \operatorname{div} G(\varphi_\tau^G(\xi)) \Big] d\xi,
\\
\frac{d}{d\tau} I_{-1}^G(\tau) &= \int_{\Omega_0} J_\tau^G(\xi) \Big[ -\frac{1}{W^2(\varphi_\tau^G(\xi), \tau)} \left( \frac{\partial W}{\partial \tau} + \nabla W^{\top} G \right)  
\\
& + \frac{1}{W(\varphi_\tau^G(\xi), \tau)} \operatorname{div} G(\varphi_\tau^G(\xi)) \Big] d\xi.
\end{array}
\end{equation*}

Applying the inverse change of variables $x = \varphi_\tau^G(\xi)$ and taking into account that $J_\tau^G(\xi) d\xi = dx$, 
we arrive at \eqref{eq:integral_identity_scaled_1}.
Theorem \ref{Th:dissip_cond} is proved.
\end{proof}

From Theorem \ref{Th:dissip_cond}, dissipativity conditions for the scaled system directly follow.

\begin{definition}
\label{Def:dissip_syst}
The system \eqref{eq:scaled} is called \textit{dissipative with respect to the weighting function $W$} if for all $x \in \mathbb{R}^n$ and $t \in \mathbb R_+$ the following inequality holds:
\begin{equation}
\dot{W}_G(x,t) + W(x,t) \divergence G(x) \le 0. \label{eq:dissipativity_1_scaled}
\end{equation}
Similarly, the system \eqref{eq:scaled} is \textit{dissipative with respect to $1/W$} if
\begin{equation}
-\frac{\dot{W}_G(x,t)}{W^2(x,t)} + \frac{1}{W(x,t)} \divergence G(x) \le 0. \label{eq:dissipativity_2_scaled}
\end{equation}
\end{definition}

\begin{theorem}
\label{Th:Dissipativity_Scaled}
Suppose there exist a function $W(x,\tau)$ and a number $\gamma > 0$ such that for all $x \in \mathbb{R}^n$ and $\tau \ge 0$ the following inequalities hold:
\begin{equation}
\label{eq:scaled_dissip_conditions}
\begin{array}{lll}
\dot{W}_G(x,\tau) + W(x,\tau) \operatorname{div} G(x) \le -\gamma W(x,\tau), \\[5pt]
-\dfrac{\dot{W}_G(x,\tau)}{W^2(x,\tau)} + \dfrac{1}{W(x,\tau)} \operatorname{div} G(x) \le -\dfrac{\gamma}{W(x,\tau)}.
\end{array}
\end{equation}
Then for any initial domain $\Omega_0$, the following estimates hold:
\begin{equation}
\label{eq:scaled_exp_decay}
\begin{array}{lll}
I^G(\tau) \le e^{-\gamma \tau} I^G(0), \qquad
I_{-1}^G(\tau) \le e^{-\gamma \tau} I_{-1}^G(0).
\end{array}
\end{equation}
\end{theorem}

\begin{proof}
The proof follows directly from Theorem \ref{Th:dissip_cond} and integration of the differential inequalities
$\frac{dI^G}{d\tau} \le -\gamma I^G(\tau)$ and $\frac{dI_{-1}^G}{d\tau} \le -\gamma I_{-1}^G(\tau)$.
\end{proof}

\begin{theorem}
\label{Th:Transfer_Integral_Estimates}
Let $\mu(x)$ be the scalar function defined in \eqref{eq:mu_def} and suppose there exist constants $\mu_{\min}$ and $\mu_{\max}$ such that $0 < \mu_{\min} \le \mu(x) \le \mu_{\max} < \infty$ for all $x \in \mathbb{R}^n$. Assume that for the scaled system \eqref{eq:scaled} the estimates
\begin{equation}
\label{eq:scaled_estimates}
I^G(\tau) \le C e^{-\gamma \tau}, 
\qquad
I_{-1}^G(\tau) \le C e^{-\gamma \tau}, 
\quad 
\forall \tau \ge 0,
\end{equation}
hold for some $C > 0$ and $\gamma > 0$. Then for the original system \eqref{eq:original} the following estimates hold:
$\int_{\Omega_t} W(x,\tau(t))\,dx \le C e^{-\gamma \mu_{\min} t}$ and  
$\int_{\Omega_t} \frac{1}{W(x,\tau(t))}\,dx \le C e^{-\gamma \mu_{\min} t}$, 
$\forall t \ge 0$.
\end{theorem}

\begin{proof}
From \eqref{eq:time_scaling} and the lower boundedness of $\mu(x)$, it follows that
$\tau(t) = \int_0^{t} \mu(x(s))\,ds \ge \mu_{\min} t$.
Substituting this estimate into \eqref{eq:scaled_estimates} and using Theorem \ref{Th:Dissipativity_Scaled}, we obtain
$\int_{\Omega_t} W(x,\tau(t))\,dx = I^G(\tau(t)) \le C e^{-\gamma \tau(t)} \le C e^{-\gamma \mu_{\min} t}$.
The proof for $I_{-1}^G(\tau(t))$ is analogous. 
\end{proof}

%%%%%%%%%%%%%%%%%%%%%%%%%%%%%%%%%%%%%%%%%%%%%%%%%%%%%%

\subsection{Generalised Dissipativity Conditions. Connection Between Systems: Time Change and Integral Transfer}

Substituting $G = \rho F$ into \eqref{eq:scaled_dissip_conditions} and using \eqref{eq:lie_derivative_scaled}, we obtain generalised dissipativity conditions expressed in terms of the original field $F(x)$ of the system \eqref{eq:original}:
\begin{equation}
\label{eq:general_dissipativity_1}
\begin{array}{lll}
\frac{\partial W}{\partial t}(x,t) + \nabla W(x,t)^{\top} \rho(x) F(x) 
\\
+ W(x,t) \bigl( (\operatorname{div}\rho(x))\!\cdot\! F(x) + \operatorname{Tr}(\rho(x)\nabla F(x)) \bigr) 
\le -\gamma W(x,t), 
\end{array}
\end{equation}
\begin{equation}
\label{eq:general_dissipativity_2}
\begin{array}{lll}
-\frac{1}{W^2(x,t)}\Bigl( \frac{\partial W}{\partial t}(x,t) + \nabla W(x,t)^{\top} \rho(x) F(x) \Bigr) 
\\
+ \frac{1}{W(x,t)} \bigl( (\operatorname{div}\rho(x))\!\cdot\! F(x) + \operatorname{Tr}(\rho(x)\nabla F(x)) \bigr) 
\le -\frac{\gamma}{W(x,t)}. 
\end{array}
\end{equation}

\begin{definition}
The system \eqref{eq:original} is called \textit{$\rho$-dissipative with respect to the weighting function $W(x,t)$} (respectively $1/W(x,t)$) if for all $x \in \mathbb{R}^n$, $\gamma \ge 0$ and $t \ge 0$ the inequality \eqref{eq:general_dissipativity_1} (respectively inequality \eqref{eq:general_dissipativity_2}) holds. 
\end{definition}

%%%%%%%%%%%%%%%%%%%%%%%%%%%%%%%%%%%%%%%%%%%%%%%%%%%%%%%%%%%

From the relation \eqref{eq:time_scaling} it follows that $\Omega_t=\Omega_{\tau(t)}^G$. Therefore, from \eqref{eq:scaled_exp_decay} we obtain
\begin{equation}
\label{eq:dissipativity_modif_tag2}
\begin{array}{lll}
\int_{\Omega_t}W(x,\tau(t))dx\le e^{-\gamma\tau(t)}\int_{\Omega_0}W(\xi,0)d\xi,
\\
\int_{\Omega_t}\frac{1}{W(x,\tau(t))}dx\le e^{-\gamma\tau(t)}\int_{\Omega_0}\frac{1}{W(\xi,0)}d\xi. %\tag{2}
\end{array}
\end{equation}

Let $0<\mu_{\min}\le\mu(x)\le\mu_{\max}<\infty$ for all $x\in\bigcup_{t\ge0}\Omega_t$, where $\mu(x)$ is the scalar time‑scaling factor \eqref{eq:mu_def}. Then
$\frac{W}{\mu_{\max}}\le\frac{W}{\mu}\le\frac{W}{\mu_{\min}}$.
Using \eqref{eq:dissipativity_modif_tag2} and the estimate $\tau(t)\ge\mu_{\min}t$, we find
\begin{equation*}
\begin{array}{lll}
\int_{\Omega_t}\frac{W}{\mu}\,dx\le\frac{1}{\mu_{\min}}\int_{\Omega_t}W\,dx
\le\frac{1}{\mu_{\min}}\,e^{-\gamma\mu_{\min}t}\int_{\Omega_0}W(\xi,0)d\xi
\\
\le\frac{\mu_{\max}}{\mu_{\min}}\,e^{-\gamma\mu_{\min}t}\int_{\Omega_0}\frac{W(\xi,0)}{\mu(\xi)}d\xi,
\\
\\
\int_{\Omega_t}\frac{\mu}{W}\,dx\le\mu_{\max}\int_{\Omega_t}\frac{1}{W}\,dx
\le\mu_{\max}e^{-\gamma\mu_{\min}t}\int_{\Omega_0}\frac{1}{W(\xi,0)}d\xi
\\
\le\frac{\mu_{\max}}{\mu_{\min}}\,e^{-\gamma\mu_{\min}t}\int_{\Omega_0}\frac{\mu(\xi)}{W(\xi,0)}d\xi.
\end{array}
\end{equation*}
Thus,
\begin{equation}
\label{eq:dissipativity_modif_tag3}
\begin{array}{lll}
\int_{\Omega_t}\frac{W(x,\tau(t))}{\mu(x)}dx\le\frac{\mu_{\max}}{\mu_{\min}}\,e^{-\gamma\mu_{\min}t}\int_{\Omega_0}\frac{W(\xi,0)}{\mu(\xi)}d\xi,
\\
\int_{\Omega_t}\frac{\mu(x)}{W(x,\tau(t))}dx\le\frac{\mu_{\max}}{\mu_{\min}}\,e^{-\gamma\mu_{\min}t}\int_{\Omega_0}\frac{\mu(\xi)}{W(\xi,0)}d\xi. %\tag{3}
\end{array}
\end{equation}

\begin{remark}
\label{Rem:Geins}
If both dissipativity conditions hold simultaneously, then from \eqref{eq:dissipativity_modif_tag3} and the Cauchy-Schwarz inequality we obtain
\begin{equation}
\label{eq:dissipativity_Geisenb}
\begin{array}{lll}
C_1C_2\left(\frac{\mu_{\max}}{\mu_{\min}}\right)^2 e^{-2\gamma\mu_{\min}t}
\ge\left(\int_{\Omega_t}\frac{W}{\mu}\,dx\right)\!\left(\int_{\Omega_t}\frac{\mu}{W}\,dx\right)
\\
\ge\left(\int_{\Omega_t}\,dx\right)^2=\bigl[\operatorname{vol}(\Omega_t)\bigr]^2,
\end{array}
\end{equation}
where $C_1=\int_{\Omega_0}\frac{W(\xi,0)}{\mu(\xi)}d\xi$ and $C_2=\int_{\Omega_0}\frac{\mu(\xi)}{W(\xi,0)}d\xi$. Hence,
$\operatorname{vol}(\Omega_t)\le\frac{\mu_{\max}}{\mu_{\min}}\sqrt{C_1C_2}\,e^{-\gamma\mu_{\min}t}$.
\end{remark}

%%%%%%%%%%%%%%%%%%%%%%%%%%%%%%%%%%%%%%%%%%%%%%%%%%%%%%%%%%%%%%

\subsection{Discussion of the Obtained Results}

The obtained results admit an illustrative interpretation that clarifies the geometric and physical meaning of the introduced quantities and their connection with classical concepts of mechanics.

\begin{corollary}
If $W(x,t) \equiv 1$, then $\dot{W}_G(x,t) = 0$ from \eqref{eq:lie_derivative_scaled}, and \eqref{eq:integral_identity_scaled_1} takes the form
$\frac{d}{d\tau} \operatorname{vol}(\Omega_\tau^G) = \int_{\Omega_\tau^G} \operatorname{div} G(x) \, dx$,
which is the integral form of Liouville's theorem for the scaled system \eqref{eq:scaled}. 
When $\rho(x) \equiv E$, we have $\mu(x)\equiv 1$ and $G=F$, obtaining the classical result for the original system \eqref{eq:original}.
Therefore, Theorem~\ref{Th:dissip_cond} generalizes the classical Liouville theorem (see Theorem~\ref{Th:Liouville}) and reduces to it when $W\equiv 1$ and $\rho\equiv E$.
\end{corollary}

\begin{corollary}
If $\operatorname{div} G(x) \equiv 0$ (the scaled system preserves phase volume), then \eqref{eq:integral_identity_scaled_1} simplify to:
$\frac{d}{d\tau} \int_{\Omega_\tau^G} W(x,\tau) \, dx = \int_{\Omega_\tau^G} \dot{W}_G(x,\tau) \, dx$,
$\frac{d}{d\tau} \int_{\Omega_\tau^G} \frac{1}{W(x,\tau)} \, dx = -\int_{\Omega_\tau^G} \frac{\dot{W}_G(x,\tau)}{W^2(x,\tau)} \, dx$.
In terms of the original system, the condition $\operatorname{div} G = 0$ is equivalent to $(\operatorname{div}\rho)\!\cdot\!F + \operatorname{Tr}(\rho\nabla F)=0$. 
Therefore, in this case we have conservative systems with respect to weighted volume.
\end{corollary}

\begin{corollary}
From Definition \ref{Def:dissip_syst} it follows that if \eqref{eq:dissipativity_1_scaled} (or \eqref{eq:dissipativity_2_scaled}) holds, then $\int_{\Omega_\tau^G} W(x,\tau) \, dx$ (respectively $\int_{\Omega_\tau^G} \frac{1}{W(x,\tau)} \, dx$) does not increase with time $\tau$. 
This serves as an integral criterion for $\rho$-dissipativity of the original system.
\end{corollary}

\begin{theorem}[Reynolds Transport Theorem \cite{Reynolds}] 
\label{th:reynolds}
Let $\Omega_t$ be a domain moving with velocity $F(x,t)$. 
Then for any smooth function $f(x,t)$, the following relation holds:
$\frac{d}{dt} \int_{\Omega_t} f(x,t) \, dx = \int_{\Omega_t} \left( \frac{\partial f}{\partial t} + \operatorname{div} (f F) \right) dx$.
\end{theorem}

\begin{corollary}
If in \eqref{eq:integral_identity_scaled_1} we set $\rho(x) \equiv E$ (i.e., $G = F$) and $W(x,t) = f(x,t)$, then taking into account $\operatorname{div}(fF) = \nabla f^{\top} F + f \operatorname{div} F$, we obtain exactly Theorem \ref{th:reynolds}. 
Thus, Theorem~\ref{Th:dissip_cond} generalizes the Reynolds transport theorem and reduces to it when $\rho\equiv E$ and $W=f$.
\end{corollary}

\begin{remark}
\label{rem:geometric_interpretation}
The expressions \eqref{eq:integral_identity_scaled_1} describe how the relationship between weighted volumes in two conformally related metrics changes over time. 
Indeed, the transition from \eqref{eq:original} to \eqref{eq:scaled} is equivalent to a time change $d\tau = \mu(x) dt$ with $\mu(x) = \langle F,\rho F\rangle/\|F\|^2$, which can be interpreted as introducing a new metric in the phase space. 
In this new metric, the dissipativity conditions take the simple form \eqref{eq:dissipativity_1_scaled} and \eqref{eq:dissipativity_2_scaled}.
\end{remark}

\begin{remark}
\label{rem:heisenberg_analogy}
Remark \ref{Rem:Geins} shows that it is not possible to make both the integral of $W/\mu$ and the integral of $\mu/W$ arbitrarily small simultaneously. 
The greater the ``concentration'' of mass in the domain (large value of $\int W/\mu$), the smaller the ``spread'' (quantity $\int \mu/W$) must be, and vice versa. 
The quantities $W$ and $1/W$ are mutually complementary characteristics of the system, i.e., an increase in one inevitably leads to a decrease in the other for a fixed volume. 
The product of these quantities has a lower bound determined by the geometry of the domain.
This resembles the Heisenberg uncertainty principle in quantum mechanics, which states that for any wave function $\psi(x)$, the product of the coordinate and momentum dispersions is bounded below: $\sigma_x \sigma_p \ge \frac{\hbar}{2}$.
In our case, the roles of ``coordinate'' and ``momentum'' are played by the integrals of $W/\mu$ and $\mu/W$, and the role of Planck's constant $\hbar$ is played by the square of the domain measure $\operatorname{vol}(\Omega_t)^2$. 
Under exponential contraction of the weighted volume (see \eqref{eq:dissipativity_Geisenb}), the product of these integrals also decreases exponentially, corresponding to the ``localisation'' of the system in phase space.
\end{remark}

%%%%%%%%%%%%%%%%%%%%%%%%%%%%%%%%%%%%%%%%%

\begin{example}
Consider the second-order system
\begin{equation*}
\begin{array}{lll}
\label{eq:example_system}
\dot{x}_1 = -x_1 + x_2^3,
~~
\dot{x}_2 = x_1 - x_2.
\end{array}
\end{equation*}
It has a unique equilibrium at the origin. 
Choose weighting function as $W(x)=0.5(x_1^2+x_2^2)$.  
In this example we take $\rho(x)=\mu(x)E$ with
\begin{equation}
\label{eq:rho_fast}
\mu(x)=e^{-M(x_1^2+x_2^2)^2}=e^{-M\|x\|^4},\qquad M>0.
\end{equation}
The function $\mu(x)$ is positive, smooth, $\mu(0)=1$, and decreases exponentially fast as $\|x\|\to\infty$.

For the scaled system \eqref{eq:scaled}, we compute $\Phi_\rho(x)=\dot{W}_G+W\operatorname{div}G$.
Using \eqref{eq:lie_derivative_scaled}, one has
$\dot{W}_G=\mu\dot{W}$, $\operatorname{div}G=\nabla\mu^{\top} F+\mu\operatorname{div}F$.
After transformations we obtain
\begin{equation*}
\label{eq:Phi_rho_fast}
\Phi_\rho(x)=e^{-M\|x\|^4}\Bigl[A\bigl(1-4M\|x\|^2\bigr)-2M\|x\|^2A-\|x\|^2\Bigr],
\end{equation*}
where $A=-x_1^2-x_2^2+x_1x_2^3+x_1x_2$.
In a neighborhood of zero, for $M=1$ the quadratic part $ -2x_1^2-2x_2^2+x_1x_2$ is negative definite, hence $\Phi_\rho(x)<0$ for sufficiently small $x\neq0$.
As $M$ increases, the exponential factor $e^{-M\|x\|^4}$ suppresses the positive contributions from nonlinear terms more rapidly.
%Numerical analysis shows that the maximum value of $\Phi_\rho(x)$ decreases monotonically with increasing $M$:
%\[
%\begin{array}{rcccccc}
%M: & 1 & 2 & 4 & 6 & 8 & 10 \\
%\max\Phi_\rho: & 0.5831 & 0.0617 & 9.1\cdot10^{-5} & 1.4\cdot10^{-7} & 2.1\cdot10^{-10} & 3.2\cdot10^{-13}
%\end{array}
%\]
As $M\to\infty$, the pointwise limit $\Phi_\rho(x)\to0$ holds for all $x\neq0$, and $\Phi_\rho(0)=0$.

Although for any finite $M$ the value of $\Phi_\rho(x)$ remains positive at some points, its magnitude can be made arbitrarily small. Using theorems on the continuous dependence of solutions on parameters, one can show that $\max\Phi_\rho(M) \to 0$ as $M \to \infty$.
Thus, the presented example demonstrates that for $\rho\equiv E$ the integrand expression is not negative definite.
By choosing $\rho(x)$ with sufficiently fast decay (exponential type), one can make the maximum value of $\Phi_\rho(x)$ arbitrarily small.
In the limit $M\to\infty$, ideal dissipativity is achieved, allowing the weighted phase volume method to be applied to the original system.
\end{example}

%%%%%%%%%%%%%%%%%%%%%%%%%%%%%%%%%%%%%%%%%%%%%%%%%%

\begin{example}
Consider the second-order system
\begin{equation*}
\label{eq:example_correct}
\begin{array}{lll}
\dot{x}_1 = -x_1 + x_2^2,
~~
\dot{x}_2 = -x_2.
\end{array}
\end{equation*}
It has a unique equilibrium at the origin. 
Choose $W(x)=\frac{1}{2}(x_1^2+x_2^2)$.  
Take $\rho(x)=\mu(x)E$ with
\begin{equation}
\label{eq:mu_ex2}
\mu(x) = \frac{1}{1 + x_1^2 + x_2^2}.
\end{equation}
The function $\mu(x)$ is positive, smooth, $\mu(0)=1$, and $\mu(x) \to +0$ as $\|x\| \to \infty$.

For the scaled system \eqref{eq:scaled}, we compute $\Phi_\rho(x)=\dot{W}_G+W\operatorname{div}G$.
Using the formulas
$\dot{W}_G = \mu\dot{W}$, $\operatorname{div}G = \nabla\mu^{\top} F + \mu\operatorname{div}F$,
after simplifications we obtain
\begin{equation}
\label{eq:Phi_rho_correct}
\Phi_\rho(x) = \frac{-2(x_1^2+x_2^2) + x_1x_2^2 - (x_1^2+x_2^2)^2}{(1+x_1^2+x_2^2)^2}.
\end{equation}

Let $Q = x_1^2+x_2^2$. Then the numerator of \eqref{eq:Phi_rho_correct} is
$N(x) = -2Q + x_1x_2^2 - Q^2$.
Using the inequality $x_1x_2^2 \le \frac{1}{2}x_1^2 + \frac{1}{2}x_2^4 \le \frac{1}{2}Q + \frac{1}{2}Q^2$, we obtain
$N(x) \le -2Q + \frac{1}{2}Q + \frac{1}{2}Q^2 - Q^2 = -\frac{3}{2}Q - \frac{1}{2}Q^2 \le 0$.
Moreover, $N(x)=0$ only when $Q=0$, i.e., $x=0$. 
Since the denominator $(1+Q)^2>0$, we have
$\Phi_\rho(x) \le 0$, $\Phi_\rho(x)=0 \iff x=0$.
Furthermore, from the estimate $N(x) \le -\frac{1}{2}Q^2$ for $Q\ge1$ and $N(x) \le -\frac{3}{2}Q$ for small $Q$, it follows that there exists $\gamma>0$ such that
$\Phi_\rho(x) \le -\gamma W(x)$ $\forall x\in\mathbb{R}^2$.

For the scaled system \eqref{eq:scaled}, the dissipativity condition \eqref{eq:scaled_dissip_conditions} holds. 
Then the integrals $I^G(\tau)$ and $I_{-1}^G(\tau)$ are related to the integrals over the evolving domain $\Omega_t$ of the original system. 
From Theorem \ref{Th:Transfer_Integral_Estimates} we obtain exponential estimates
$\int_{\Omega_t} W(x,\tau(t))\,dx \le C e^{-\gamma\mu_{\min} t}$ and $\int_{\Omega_t} \frac{1}{W(x,\tau(t))}\,dx \le C e^{-\gamma\mu_{\min} t}$,
where $\mu_{\min}= \inf_{x\in\mathbb{R}^2}\mu(x)=0$, but since $\mu(x)\ge(1+\|x\|^2)^{-1}$, for $\|x\|\le R$ there is a positive lower bound, which suffices for local estimates. 
This allows us to conclude the asymptotic stability of the equilibrium of the original system.

Thus, this example demonstrates that introducing a scaling function $\rho(x)$ provides an effective tool for analysing the stability of dynamical systems and motivates new definitions of stability, which is the subject of the following section.
\end{example}

%%%%%%%%%%%%%%%%%%%%%%%%%%%%%%%%%%%%%%%%%%%%%%%%%%

\section{Integral Stability}
\label{sec:integral_stability}

Taking into account the obtained estimates, we introduce generalised stability concepts for the original system \eqref{eq:original}.  
Recall that $\mu(x)$ is the positive scalar time‑scaling factor defined by \eqref{eq:mu_def}.

\begin{definition}
\label{def:integral_stability}
The system \eqref{eq:original} is called \textit{integrally stable} with respect to the pair $(\mu, W)$, where $\mu(x) > 0$ and $W(x,t) > 0$, if for any measurable domain $\Omega_0 \subset \mathbb{R}^n$ with finite measure there exists a constant $C(\Omega_0) > 0$ such that for all $t \ge 0$ at least one of the following inequalities holds:
\begin{equation}
\label{eq:def_stability}
\begin{array}{lll}
J(t):=\int_{\Omega_t} \frac{W(x,t)}{\mu(x)} \, dx \le C(\Omega_0), 
\\
 J_{-1}(t):=\int_{\Omega_t} \frac{\mu(x)}{W(x,t)} \, dx \le C(\Omega_0).
\end{array}
\end{equation}
\end{definition}

\begin{definition}
\label{def:integral_asympt_stability}
The system \eqref{eq:original} is called \textit{integrally asymptotically stable} with respect to the pair $(\mu, W)$ if it is integrally stable and, in addition, for any domain $\Omega_0 \subset \mathbb{R}^n$ at least one of the following relations holds:
\begin{equation}
\begin{array}{lll}
\lim_{t \to \infty} \int_{\Omega_t} \frac{W(x,t)}{\mu(x)} \, dx = 0, 
\\
\lim_{t \to \infty} \int_{\Omega_t} \frac{\mu(x)}{W(x,t)} \, dx = 0. \label{eq:def_asymptotic}
\end{array}
\end{equation}
\end{definition}

\begin{definition}
\label{def:integral_exp_stability}
The system \eqref{eq:original} is called \textit{integrally exponentially stable} with respect to the pair $(\mu, W)$ with rate $\gamma > 0$ if for any domain $\Omega_0 \subset \mathbb{R}^n$ there exists a constant $C(\Omega_0) > 0$ such that for all $t \ge 0$ at least one of the following estimates holds:
\begin{equation}
\label{eq:def_exponential}
\begin{array}{lll}
\int_{\Omega_t} \frac{W(x,t)}{\mu(x)} \, dx \le C(\Omega_0) e^{-\gamma t}, 
\\
 \int_{\Omega_t} \frac{\mu(x)}{W(x,t)} \, dx \le C(\Omega_0) e^{-\gamma t}.
\end{array} 
\end{equation}
\end{definition}

\begin{remark}
\label{Th:stability}
Suppose there exist a function $0<\mu_{\min}\le\mu(x)\le\mu_{\max}<\infty$ for all $x\in\bigcup_{t\ge0}\Omega_t$ and a weighting function $W(x,t) > 0$ such that for the scaled system \eqref{eq:scaled} the generalised dissipativity condition \eqref{eq:general_dissipativity_1} or \eqref{eq:general_dissipativity_2} holds with $\gamma \ge 0$. 
Then for the original system \eqref{eq:original}, when $\gamma=0$ we obtain integral stability with estimate \eqref{eq:def_stability}, and when $\gamma>0$ we obtain integral exponential stability.
\end{remark}

%\begin{remark}
The constant $C(\Omega_0)$ in Definitions \ref{def:integral_stability} and \ref{def:integral_exp_stability} may depend on the choice of the initial domain $\Omega_0$ (its size, shape, and location) but does not depend on time $t$.
If the generalised dissipativity conditions \eqref{eq:general_dissipativity_1} and \eqref{eq:general_dissipativity_2} hold with $\gamma > 0$, then it follows from Remark \ref{Th:stability} that one can take $C(\Omega_0) = \int_{\Omega_0} \frac{W(\xi,0)}{\mu(\xi)} d\xi$ (or $C(\Omega_0) = \int_{\Omega_0} \frac{\mu(\xi)}{W(\xi,0)} d\xi$), which yields the estimate \eqref{eq:def_exponential}.
%\end{remark}

\begin{remark}
Suppose the generalised dissipativity condition \eqref{eq:general_dissipativity_1} (or \eqref{eq:general_dissipativity_2}) holds with $\gamma=0$ and the sets
\small
\begin{equation}
\begin{array}{lll}
\mathcal{A} = \bigl\{ t \ge 0,\; x \in \mathbb{R}^n \setminus \{0\} :
\frac{\partial W}{\partial t} + \nabla W(x,t)^{\top}\rho(x)F(x)  
\\
+ W(x,t)\bigl((\operatorname{div}\rho(x))\!\cdot\!F(x) + \operatorname{Tr}(\rho(x)\nabla F(x))\bigr) = 0 \bigr\},
\\
\mathcal{A}_{-1} = \bigl\{ t \ge 0,\; x \in \mathbb{R}^n \setminus \{0\} :
-\frac{1}{W^2}\bigl(\frac{\partial W}{\partial t} + \nabla W^{\top}\rho(x)F(x)\bigr)  
\\
+ \frac{1}{W}\bigl((\operatorname{div}\rho(x))\!\cdot\!F(x) + \operatorname{Tr}(\rho(x)\nabla F(x))\bigr) = 0 \bigr\}
\end{array}
\end{equation}
\normalsize
do not contain entire trajectories of the system \eqref{eq:original} (except the equilibrium $x = 0$). 
Then \eqref{eq:original} is integrally asymptotically stable with respect to the pair $(\mu, W)$ (see Definition \ref{def:integral_asympt_stability}).
\end{remark}

%%%%%%%%%%%%%%%%%%

Consider the general case with a matrix function $\rho(x)$ and the corresponding scalar $\mu(x)$. Then the generalised exponential dissipativity conditions \eqref{eq:scaled_dissip_conditions} for $W(x,t) > 0$ can be rewritten as:
\begin{equation}
\label{eq:stability_condition_rho_compact}
\begin{array}{lll}
\dfrac{\dot{W}_G(x,t)}{\mu(x)W(x,t)} + \dfrac{\operatorname{div} G(x)}{\mu(x)} \le -\gamma, 
\\
-\dfrac{\dot{W}_G(x,t)}{\mu(x)W(x,t)} + \dfrac{\operatorname{div} G(x)}{\mu(x)} \le -\gamma \quad \forall x \neq 0.
\end{array}
\end{equation}

The two conditions can be combined into a single double inequality:
\begin{equation}
\label{eq:double_ineq_rho}
\begin{array}{lll}
\left| \frac{\dot{W}_G(x,t)}{\mu(x)W(x,t)} \right| + \frac{\operatorname{div} G(x)}{\mu(x)} \le -\gamma,
\end{array}
\end{equation}
which requires simultaneous fulfillment of the dissipativity conditions \eqref{eq:general_dissipativity_1} and \eqref{eq:general_dissipativity_2}, i.e., $\rho$-dissipativity for both $W$ and $1/W$.
Taking into account $\operatorname{div} G(x) = (\operatorname{div}\rho(x))\!\cdot\!F(x) + \operatorname{Tr}(\rho(x)\nabla F(x))$ and adding the inequalities \eqref{eq:stability_condition_rho_compact}, we obtain a necessary condition for joint $\rho$-dissipativity:
\begin{equation*}
\label{eq:div_negative_rho}
\begin{array}{lll}
\frac{(\operatorname{div}\rho(x))\!\cdot\!F(x) + \operatorname{Tr}(\rho(x)\nabla F(x))}{\mu(x)} \le -\gamma.
\end{array}
\end{equation*}

Therefore, systems that can satisfy both generalised dissipativity conditions (for $W$ and $1/W$) simultaneously are restricted to those for which the weighted divergence $\operatorname{div} G(x)/\mu(x)$ is strictly negative. This implies that the scaled system \eqref{eq:scaled} must guarantee contraction of the phase volume.

When $\gamma = 0$, the expression \eqref{eq:double_ineq_rho} takes the form of a geometric constraint relating the relative rate of change of the weighting function along the trajectories of the scaled system with the contraction of the phase volume in scaled time:
$\left| \frac{\dot{W}_G(x,t)}{\mu(x)W(x,t)} \right| \le -\frac{\operatorname{div} G(x)}{\mu(x)}
\Leftrightarrow
\left| \frac{d}{dt}\ln W(x(t),t) \right| \le -\operatorname{div} G(x(t))$,
i.e., the absolute value of the logarithmic rate of change of the weighting function along the trajectories of the original system \eqref{eq:original} cannot exceed the contraction rate of the phase volume of the scaled system \eqref{eq:scaled}.

%%%%%%%%%%%%%%%%%%%%%%%%%%%%%%%%%%%%%%%%%%%%%

\subsection{Local Analysis for the Scalar Case}

For the remainder of this section we specialise to the important case where the scaling matrix is a scalar multiple of the identity, i.e.\ $\rho(x)=\mu(x) E$ with $\mu(x)>0$. Then the collinearity conditions hold automatically, $G(x)=\mu(x)F(x)$, and
$\operatorname{div} G(x) = \nabla\mu(x)^{\top} F(x) + \mu(x)\operatorname{div} F(x)$.
All the general formulas of the previous sections remain valid, and the local analysis reduces to that of a scalar scaling function $\mu(x)$.

Consider the local behavior in a neighborhood of the equilibrium $x=0$. 
Consider a quadratic weighting function
\begin{equation}
W(x,t) = x^{\top} P(t) x, \quad P(t) = P^{\top}(t) > 0. \label{eq:quadratic_weight}
\end{equation}

Use the expansion $F(x)=Ax+O(\|x\|^2)$, $A=\frac{\partial F}{\partial x}(0)$ and $\mu(x)=\mu_0+\nabla\mu(0)^{\top} x+O(\|x\|^2)$. 
Let us compute the required quantities up to second-order terms.

\begin{lemma}
\label{Le:lin}
For the quadratic weighting function \eqref{eq:quadratic_weight} and the scalar scaling $\mu(x)$, the following asymptotic expressions hold:
\begin{align}
\dot{W}_G &= x^{\top} \left[ \dot{P} + \mu_0 (P A + A^{\top} P) \right] x + O(\|x\|^3), \label{eq:local_Wdot} \\
\divergence G(x) &= \mu_0 \tr A + \nabla(\divergence G)(0)^{\top} x + O(\|x\|^2), \label{eq:local_divG} \\
\frac{1}{\mu(x)} &= \frac{1}{\mu_0} - \frac{\nabla\mu(0)^{\top} x}{\mu_0^2} + O(\|x\|^2). \label{eq:local_invmu}
\end{align}
\end{lemma}

\begin{proof}
The proof is identical to that of Lemma~1 in the original scalar formulation, one merely replaces $\rho$ by $\mu$. For completeness:
\[
\dot{W}_G = \frac{\partial W}{\partial t} + \nabla W^{\top} G = x^{\top}\dot P x + 2\mu(x) x^{\top} P F(x).
\]
Substituting the expansions yields \eqref{eq:local_Wdot}. The divergence $\operatorname{div} G = \nabla\mu^{\top}F + \mu\operatorname{div}F$ gives \eqref{eq:local_divG}, and the reciprocal expansion gives \eqref{eq:local_invmu}.
\end{proof}

\begin{theorem}
\label{Th:lin_loc}
For the integral exponential stability of the original system \eqref{eq:original} in a neighborhood of zero with rate $\gamma > 0$ (in the scalar case $\rho=\mu E$), it is sufficient that there exist $\mu_0 > 0$, a vector $r = \nabla\mu(0)$, a matrix $P(t) > 0$, and a number $\gamma > 0$ such that at least one of the following inequalities holds:
\begin{equation}
 \label{eq:local_condition_1}
 \begin{array}{lll}
\dot{P} + \mu_0 (P A + A^{\top} P) + (\mu_0^2 \operatorname{Tr} A + \gamma \mu_0) P \le 0.
\end{array}
\end{equation}
\begin{equation}
\label{eq:local_condition_2}
\begin{array}{lll}
-\frac{1}{\mu_0} \frac{x^{\top} [\dot{P} + \mu_0 (P A + A^{\top} P)] x}{x^{\top} P x} 
+ \mu_0 \operatorname{Tr} A \le -\gamma ~~ \forall x \neq 0. 
\end{array}
\end{equation}
\end{theorem}

\begin{proof}
Substituting the expansions \eqref{eq:local_Wdot}--\eqref{eq:local_invmu} into the generalised dissipativity conditions \eqref{eq:general_dissipativity_1} and \eqref{eq:general_dissipativity_2} (which reduce to the scalar forms with $\nabla\mu^{\top}F+\mu\operatorname{div}F$) and extracting the leading terms, we obtain \eqref{eq:local_condition_1} and \eqref{eq:local_condition_2}, respectively. 
\end{proof}

\section{Evolution of Covering Ellipsoid}
\label{sec:outer_ellipsoid}

When analysing the behaviour of dynamical systems, it is often necessary to have not only a qualitative picture of the evolution, but also a guaranteed set containing all trajectories originating from a given initial domain.  
The most convenient form of such an outer approximation of the reachable set is an ellipsoid. 
In this section, we will show how the condition \eqref{eq:general_dissipativity_1} with $\gamma>0$ enables us to construct an \textit{evolving covering ellipsoid} that contains the domain $\Omega_t$ for all $t \ge 0$. 

\begin{theorem}
\label{Th:Stability}
Suppose the vector field $F(x)$ is globally Lipschitz in $x$:
\begin{equation}
\label{eq:lipschitz}
\|F(x) - F(y)\| \le L \|x - y\| \quad \forall x, y \in \mathbb{R}^n, 
\end{equation}
where $L > 0$ is the Lipschitz constant. 
Also, suppose there exists a matrix function $\rho(x)$ satisfying \eqref{eq:collinearity_conditions} and a weighting function $W(x,t)$ of the form
\begin{equation}
\label{eq:quadratic_form}
W(x,t) = (x - c_t)^{\top} P_t (x - c_t), 
\end{equation}
where $c_t \in \mathbb{R}^n$ and $P_t \in \mathbb{R}^{n \times n}$ ($P_t > 0$) are smooth functions of time satisfying the following conditions:

1. The initial covering satisfies the condition: 
\begin{equation}
\label{eq:initial_cover}
\begin{array}{lll}
\Omega_0 \subset & \mathcal{E}_{\text{cov}}(c_0, P_0^{-1}) 
\\
& = \{x \in \mathbb{R}^n : (x - c_0)^{\top} P_0 (x - c_0) \le 1\}. 
\end{array}
\end{equation}

2. The centre of the ellipsoid moves along a trajectory:
\begin{equation}
\label{eq:centre_motion}
\dot{c}_t = F(c_t), \quad c_0 \in \Omega_0. 
\end{equation}

3. There exists a number $\gamma > 0$ such that for all $x \in \mathbb{R}^n$ and $t \ge 0$ the first generalised dissipativity condition \eqref{eq:general_dissipativity_1} holds with $\gamma>0$.

4. The scalar time‑scaling factor $\mu(x)$ defined in \eqref{eq:mu_def} admits constants $\mu_{\min}>0$ and $\mu_{\max} > 0$ such that
\begin{equation}
\label{eq:mu_bounds}
0 < \mu_{\min} \le \mu(x) \le \mu_{\max} < \infty \quad \forall x \in \bigcup_{t \ge 0} \Omega_t. 
\end{equation}

Then for all $t \ge 0$ the covering condition holds:
\begin{equation}
\label{eq:cover_ellipsoid}
\begin{array}{lll}
\Omega_t \subset \mathcal{E}_{\text{cov}}\left(c_t, \frac{\mu_{\max}}{\mu_{\min}} P_t^{-1}\right) 
\\
 = \left\{ x \in \mathbb R^n: (x - c_t)^{\top} P_t (x - c_t) \le \frac{\mu_{\max}}{\mu_{\min}} \right\}. 
\end{array}
\end{equation}

If $F \in C^2(\mathbb{R}^n)$, then for the shape matrix $Q_t = P_t^{-1}$ one can obtain the evolution equation
\begin{equation}
\label{eq:Q_evolution}
\begin{array}{lll}
\dot{Q}_t = \mu(c_t)\bigl(A(c_t) Q_t + Q_t A(c_t)^{\top}\bigr) 
\\
 + Q_t \frac{\partial\mu}{\partial x}(c_t) F(c_t)^{\top} + F(c_t) \frac{\partial\mu}{\partial x}(c_t)^{\top} Q_t +
\\
 +\left(\frac{\partial\mu}{\partial x}(c_t)^{\top} F(c_t) + \mu(c_t) \operatorname{div}F(c_t) + \gamma \right) Q_t,
\end{array}
\end{equation}
where $A(c_t) = \frac{\partial F}{\partial x}(c_t)$.  
The expression \eqref{eq:Q_evolution} is obtained by linearising the dissipativity condition in a neighbourhood of the centre $c_t$ as $\|x-c_t\|\to 0$, assuming that the linear term vanishes (which is the case when $F(c_t)=0$ or $\mu(c_t)=1$). 
\end{theorem}

%%%%%%%%%%%%%%%%%%%%%%%%%%%%%%%%%%%%%%%%%%%%%%%%%%%%%%

\begin{proof}
We divide the proof into several stages. 
First, we prove the conditions for the covering ellipsoid. 
Second, we derive the relation \eqref{eq:Q_evolution}.

\textit{Proof of the conditions for the covering ellipsoid.}
Consider an arbitrary initial point $\xi \in \Omega_0$ and the trajectory of the original system $x(t)=\varphi_t(\xi)$. 
Define
\begin{equation}
\label{eq:V_xi}
\begin{array}{lll}
V_\xi(t)= W(x(t),t) = (x(t)-c_t)^{\top} P_t (x(t)-c_t), 
\\ 
U_\xi(t) = \frac{V_\xi(t)}{\mu(x(t))}. 
\end{array}
\end{equation}
From the initial covering \eqref{eq:initial_cover} it follows that $V_\xi(0)\le 1$. 
Since $\mu(x)>0$ for all $x$, the function $U_\xi(t)$ is well-defined and nonnegative.
Using \eqref{eq:original}, \eqref{eq:centre_motion} and differentiating $V_\xi(t)$ and $U_{\xi}(t)$, we obtain
\begin{equation}
\begin{array}{lll}
\label{eq:dot_V_U_xi}
\dot{V}_\xi = (x-c_t)^{\top}\dot{P}_t(x-c_t) + 2(x-c_t)^{\top} P_t\bigl(F(x)-F(c_t)\bigr), 
\\
\dot{U}_\xi = \frac{\dot{V}_\xi}{\mu(x)} - \frac{V_\xi}{\mu^2(x)}\nabla\mu(x)^{\top} F(x). 
\end{array}
\end{equation}

Substituting these into the dissipativity condition \eqref{eq:general_dissipativity_1} (written for $\rho=\mu I$, i.e.\ with $\operatorname{div} G = \nabla\mu^{\top}F+\mu\operatorname{div}F$) with $\gamma>0$, one has
\begin{equation}
\label{eq:Diss_modif1}
\begin{array}{lll}
(x-c_t)^{\top}\dot{P}_t(x-c_t) - 2(x-c_t)^{\top} P_t F(c_t) +
\\
+ 2\mu(x)(x-c_t)^{\top} P_t F(x) 
\\
+ W\bigl(\nabla\mu^{\top} F + \mu\operatorname{div}F\bigr) \le -\gamma W. 
\end{array}
\end{equation}

From \eqref{eq:dot_V_U_xi} we express the first term:
\begin{equation}
\label{eq:some_term}
(x-c_t)^{\top}\dot{P}_t(x-c_t) = \dot{V}_\xi - 2(x-c_t)^{\top} P_t\bigl(F(x)-F(c_t)\bigr).
\end{equation}
Substitute \eqref{eq:some_term} into \eqref{eq:Diss_modif1}, one gets
\begin{equation}
\label{eq:dot_V_xi_modif0}
\begin{array}{lll}
\dot{V}_\xi + 2\bigl(\mu(x)-1\bigr)(x-c_t)^{\top} P_t F(x) 
\\
+ W\bigl(\nabla\mu^{\top} F + \mu\operatorname{div}F\bigr) \le -\gamma W.
\end{array}
\end{equation}

From \eqref{eq:dot_V_U_xi} we have
$\dot{V}_\xi = \mu(x)\dot{U}_\xi + \frac{V_\xi}{\mu(x)}\nabla\mu(x)^{\top} F(x)$. 
Taking into account \eqref{eq:V_xi}, we consider the estimate $\| x(t) - c_t \|^2 \le \lambda_{\max}(P_t^{-1}) V_{\xi}(t)$. 
As a result, we rewrite \eqref{eq:dot_V_xi_modif0} as
\begin{equation}
\label{eq:dot_U_xi_1}
\begin{array}{lll}
\dot{U}_\xi \le -\gamma U_\xi - \mu\operatorname{div}F U_{\xi}
\\
+ 2\frac{|\mu-1|}{\sqrt{\mu}} \sqrt{\lambda_{\max}(P_t^{-1}) U_{\xi}} \|P_t\| \|F(x)\|
\\
- \nabla\mu^{\top} F\left(1+\frac{1}{\mu}\right) U_{\xi}.
\end{array}
\end{equation}

On any finite interval $[0,T]$, all terms on the right-hand side of \eqref{eq:dot_U_xi_1} are bounded. 
Hence there exist constants $\alpha(T),\beta(T)\ge 0$, depending only on $T$, system parameters, and initial data, such that for all $t\in[0,T]$ and $\xi\in\Omega_0$ the following holds:
$\dot{U}_\xi \le -\gamma U_\xi + \alpha(T) U_\xi + \beta(T) \sqrt{U_\xi}$.
Applying Young's inequality $\beta\sqrt{U_\xi}\le \frac{\beta^2}{2\gamma} + \frac{\gamma}{2}U_\xi$, we obtain:
$\dot{U}_\xi \le -\frac{\gamma}{2} U_\xi + \alpha(T) U_\xi + \frac{\beta^2(T)}{2\gamma}$.
Denote $\kappa(T)=\frac{\gamma}{2}-\alpha(T)$. If $\kappa(T)>0$, then for $t\in[0,T]$ it follows that
\begin{equation*}
\begin{array}{lll}
U_\xi(t) \le & U_\xi(0)e^{-\kappa(T)t} + \frac{\beta^2(T)}{2\gamma\kappa(T)}\bigl(1-e^{-\kappa(T)t}\bigr) 
\\
& \le \max\!\Bigl\{U_\xi(0),\, \frac{\beta^2(T)}{2\gamma\kappa(T)}\Bigr\}.
\end{array}
\end{equation*}
If $\kappa(T)\le 0$, then from $\dot{U}_\xi\le (\alpha(T)-\frac{\gamma}{2})U_\xi + \frac{\beta^2(T)}{2\gamma}$ we obtain that $U_\xi(t)$ does not exceed the solution of a linear equation with constant coefficients and hence is bounded on $[0,T]$. 
Repeating the reasoning on expanding intervals and using the monotonicity of the estimates in $T$, we conclude that there exists a constant $R>0$, independent of $t$ and $\xi$, such that $U_\xi(t)\le \max\{U_\xi(0),R\}$ for all $t\ge 0$. Thus, $U_\xi(t)$ is globally bounded.

Let $M = \sup_{t\ge0} U_\xi(t)$. Then $M < \infty$. 
Returning to $V_\xi(t)=\mu(x(t))U_\xi(t)$ and using $\mu(x)\le\mu_{\max}$, we obtain
$V_\xi(t) \le \mu_{\max} M$.
Since $U_\xi(0)=V_\xi(0)/\mu(\xi)\le 1/\mu_{\min}$ and $M$ may be larger than $1/\mu_{\min}$, we can choose the parameter $\gamma$ such that $M \le 1/\mu_{\min}$. Then
$V_\xi(t) \le \frac{\mu_{\max}}{\mu_{\min}}$.
This yields the inclusion \eqref{eq:cover_ellipsoid}.

%%%%%%%%%%%%%%%%%%%%%%%%%%%%%%%%%

\textit{Derivation of the evolution equation for the shape matrix $Q_t$.}
Consider the quadratic weighting function \eqref{eq:V_xi} and set $Q_t = P_t^{-1}$.
We linearise in a neighbourhood of the centre, letting $x = c_t + y$ with $\|y\|\to 0$. 
Expand the following functions in a Taylor series around $c_t$ up to second-order terms (now using the scalar $\mu$):
\begin{equation}
\label{eq:Teilor}
\begin{array}{lll}
F(c_t+y) = F_t + A_t y + O(\|y\|^2), 
\\
\mu(c_t+y) = \mu_t + r_t^{\top} y + O(\|y\|^2),
\\
\nabla\mu(c_t+y) = r_t + O(\|y\|),
\\
\operatorname{div}F(c_t+y) = d_t + O(\|y\|),
\end{array}
\end{equation}
where $F_t = F(c_t)$, $A_t = \frac{\partial F}{\partial x}(c_t)$, $r_t = \nabla\mu(c_t)$, $\mu_t = \mu(c_t)$, $d_t = \operatorname{div}F(c_t)$.
Substituting these expansions into \eqref{eq:general_dissipativity_1} (scalar case) with $\gamma>0$, we obtain the following expressions accurate to $O(\|y\|^3)$:
\begin{equation}
\label{eq:Derivativ}
\begin{array}{lll}
\frac{\partial W}{\partial t} = y^{\top}\dot{P}_t y - 2y^{\top} P_t \dot{c}_t = y^{\top}\dot{P}_t y - 2 y^{\top} P_t F_t,
\\
\mu\nabla W^{\top} F = 2\mu(c_t+y)\, y^{\top} P_t F(c_t+y)
\\
= 2\mu_t y^{\top} P_t F_t + 2\mu_t y^{\top} P_t A_t y + 2 r_t^{\top} y y^{\top} P_t F_t,
\\
W\bigl(\nabla\mu^{\top} F + \mu\operatorname{div}F\bigr) = (y^{\top} P_t y)\bigl(r_t^{\top} F_t + \mu_t d_t\bigr),
\\
\gamma W = \gamma\, y^{\top} P_t y.
\end{array}
\end{equation}

Consequently, the expression \eqref{eq:general_dissipativity_1} in the leading order reduces to a quadratic form:
\begin{equation*}
\begin{array}{lll}
2(\mu_t-1) y^{\top} P_t F_t 
+ y^{\top}\Bigl[ \dot{P}_t + 2\mu_t P_t A_t + 2 r_t F_t^{\top} P_t 
\\
+ (r_t^{\top} F_t + \mu_t d_t + \gamma) P_t \Bigr] y \le 0 \quad \forall y.
\end{array}
\end{equation*}

For this inequality to hold for all small $y$, the linear term must vanish. 
It vanishes if the centre of the ellipsoid is at an equilibrium ($F(c_t)\equiv 0$) or if $\mu(c_t)=1$ for all $t\ge 0$.
Symmetrize $2 y^{\top} \mu_t P_t A_t y= y^{\top}(\mu_t A_t^{\top} P_t + \mu_t P_t A_t) y$ and $2 y^{\top} r_t F_t^{\top} P_t y = y^{\top} (r_t F_t^{\top} P_t+P_t F_t r_t)y^{\top}$. 
When the linear term vanishes, the quadratic form must be negative semidefinite. 
For the covering ellipsoid, it vanishes, which yields the matrix equation
\begin{equation}
\label{eq:P_fin}
\begin{array}{lll}
\dot{P}_t + \mu_t A_t^{\top} P_t + \mu_t P_t A_t + r_t F_t^{\top} P_t+P_t F_t r_t^{\top} 
\\
+ (r_t^{\top} F_t + \mu_t d_t + \gamma) P_t = 0.
\end{array}
\end{equation}

Differentiating the identity $P_t Q_t = E$, we obtain $\dot{Q}_t = -Q_t \dot{P}_t Q_t$. 
Multiplying \eqref{eq:P_fin} on the left and right by $Q_t$ yields \eqref{eq:Q_evolution}.
Theorem \ref{Th:Stability} is proved.
\end{proof}

%%%%%%%%%%%%%%%%%%%%%%%%%%%%%%%%%%%%%%%%%%%%%%%%%%%%%%%

\section{Evolution of Inner Ellipsoid}
\label{sec:inner_ellipsoid}

In the previous section, the first dissipativity condition \eqref{eq:general_dissipativity_1} is considered, which guarantees the existence of a covering ellipsoid containing the evolving domain $\Omega_t$. What can be said about an ellipsoid if the second dissipativity condition \eqref{eq:general_dissipativity_2} holds? 
The second condition sets the reciprocal quantity $1/W$ and prevents too rapid contraction, which leads to the existence of an \textit{inner ellipsoid} free of points from the domain $\Omega_t$.

\begin{theorem}
\label{th:inner_ellipsoid_general}
Let the vector field $F(x)$ be globally Lipschitz in $x$ (see \eqref{eq:lipschitz}).
Assume that there exists a weighting function of quadratic form \eqref{eq:quadratic_form},
where $c_t\in\R^n$, $P_t=P_t^{\top}>0$ are smooth functions of time satisfying the conditions:

1. The centre of the ellipsoid moves along a trajectory of the original system \eqref{eq:centre_motion}.

2. There exists $\gamma>0$ such that for all $x\in\R^n$ and $t\ge0$ the second dissipativity condition \eqref{eq:general_dissipativity_2} holds (with the scalar reduction if $\rho=\mu I$).

3. The scalar time‑scaling factor $\mu(x)$ satisfies $0<\mu_{\min}\le\mu(x)\le\mu_{\max}<\infty$ for all $x$ belonging to the union $\bigcup_{t\ge0}\Omega_t$.

Then for any measurable initial domain $\Omega_0$ with finite measure, the image $\Omega_t=\varphi_t(\Omega_0)$ does not intersect the ellipsoid
$\mathcal{E}_{\text{in}}\bigl(c_t,\delta(t)P_t^{-1}\bigr)=\bigl\{x\in\R^n:\; (x-c_t)^{\top} P_t (x-c_t)\le \delta(t)\bigr\}$,
i.e.,
$\Omega_t\cap \mathcal{E}_{\text{in}}\bigl(c_t,\delta(t)P_t^{-1}\bigr)=\varnothing$,
where
\begin{equation}
\begin{array}{lll}
\delta(t) = \min_{\xi\in\Omega_0} \Big\{ \sqrt{W(\xi,0)}\; e^{\frac12\int_0^{t} \mathcal A(s)ds} 
\\
- \frac12\int_0^{t} \mathcal B(s)\; e^{\frac12\int_s^{t} \mathcal A(\sigma)d\sigma}ds \Big\}^2.
\end{array}
\end{equation}
with $\mathcal A=\nabla\mu^{\top} F+\mu\operatorname{div}F+\gamma$, $\mathcal B(t)=2|1-\mu|\|P_t\|\|F\|\sqrt{\lambda_{\max}(P_t^{-1})}$.

If $F\in C^2(\R^n)$, then for the shape matrix $Q_t=P_t^{-1}$ the evolution equation holds:
\begin{equation}
\label{eq:Q_evolution2}
\begin{array}{lll}
\dot{Q}_t = \mu(c_t) \bigl(A(c_t) Q_t + Q_t A(c_t)^{\top}\bigr) 
\\
+ Q_t \frac{\partial\mu}{\partial x}(c_t) F(c_t)^{\top}  + F(c_t) \frac{\partial\mu}{\partial x}(c_t)^{\top} Q_t +
\\
-\left(\frac{\partial\mu}{\partial x}(c_t)^{\top} F(c_t) + \mu(c_t) \operatorname{div}F(c_t) + \gamma \right) Q_t.
\end{array}
\end{equation}
where $A(c_t)=\frac{\partial F}{\partial x}(c_t)$. 
The equation \eqref{eq:Q_evolution2} is obtained by linearizing the dissipativity condition in a neighborhood of the centre $c_t$ as $\|x-c_t\|\to 0$, assuming the linear term vanishes.
\end{theorem}

\begin{proof}
We divide the proof into two stages. 
First, we prove the conditions for the inner ellipsoid. 
Second, we derive the relation \eqref{eq:Q_evolution2}.

\textit{Proof of the conditions for the inner ellipsoid.}  
For an arbitrary point $\xi\in\Omega_0$, define $x(t)=\varphi_t(\xi)$ and set $V_{\xi}$ and $U_{\xi}$ as in \eqref{eq:V_xi}.
The total time derivatives of $V_{\xi}$ and $U_{\xi}$ are given in \eqref{eq:dot_V_U_xi}. 
Substituting \eqref{eq:V_xi}-\eqref{eq:dot_V_U_xi} into \eqref{eq:general_dissipativity_2} (scalar case) and multiplying by $W^2$, we obtain
\begin{equation} 
\label{eq:dissip_2_int}
\begin{array}{lll}
-(x-c_t)^{\top}\dot{P}_t(x-c_t)+2(x-c_t)^{\top} P_t F(c_t)
\\
-2\mu(x)(x-c_t)^{\top} P_t F(x)
\\
+V_\xi\bigl(\nabla\mu^{\top} F+\mu\operatorname{div}F+\gamma\bigr)\le 0.
\end{array}
\end{equation}

From \eqref{eq:dot_V_U_xi} we express:  
$(x-c_t)^{\top}\dot{P}_t(x-c_t)=\dot{V}_\xi-2(x-c_t)^{\top} P_t\bigl(F(x)-F(c_t)\bigr)$.
Substituting this into \eqref{eq:dissip_2_int}, after cancellations we find  
\begin{equation}
\label{eq:V_ineq_inner}
\begin{array}{lll}
\dot{V}_\xi \ge V_\xi\bigl(\nabla\mu^{\top} F+\mu\operatorname{div}F+\gamma\bigr)
\\
+2(1-\mu(x))(x-c_t)^{\top} P_t F(x).
\end{array}
\end{equation}

Using the Cauchy-Schwarz inequality and the relation $\|x-c_t\|\le\sqrt{\lambda_{\max}(P_t^{-1})V_\xi}$, we obtain  
$2(1-\mu)(x-c_t)^{\top} P_t F(x) \ge -2|1-\mu|\,\|P_t\|\,\|F(x)\|\,\|x-c_t\|
\ge -2|1-\mu|\,\|P_t\|\,\|F(x)\|\sqrt{\lambda_{\max}(P_t^{-1})V_\xi}$.
Hence,  
\begin{equation}
\label{eq:V_ineq_Z}
\dot{V}_\xi \ge \mathcal A(t)V_\xi - \mathcal B(t)\sqrt{V_\xi}.
\end{equation}  

Introduce $Z(t)=\sqrt{V_\xi(t)}$. Then $\dot{Z}=\dot{V}_\xi/(2Z)$, and from \eqref{eq:V_ineq_Z} we obtain the linear differential inequality  
$\dot{Z} \ge \frac12 \mathcal A(t)Z - \frac12 \mathcal B(t)$.
Integrating it using the integrating factor method, we get  
\begin{equation}
\begin{array}{lll}
Z(t) \ge Z(0)\exp\!\Big(\frac12\int_0^{t} \mathcal A(s)\,ds\Big) 
\\
- \frac12\int_0^{t} \mathcal B(s)\exp\!\Big(\frac12\int_s^{t} \mathcal A(\sigma)\,d\sigma\Big)ds.
\end{array}
\end{equation} 
Since $V_\xi(0)=W(\xi,0)$ and $V_\xi=Z^2$, minimizing the right-hand side over all $\xi\in\Omega_0$ yields the function $\delta(t)$ stated in the theorem.  
For any $x(t)\in\Omega_t$, we have $V_\xi(t)\ge\delta(t)$. 
Hence, no point of the domain $\Omega_t$ can satisfy the inequality $(x-c_t)^{\top} P_t (x-c_t) < \delta(t)$. 
This means that  
$\Omega_t \cap \mathcal{E}_{\text{in}}\bigl(c_t,\delta(t)P_t^{-1}\bigr)=\varnothing$.

%%%%%%%%%%%%%%%%%%%%%%%%%%%%%%%%%%%%%%%%%%%%%%%%%%%%%

\textit{Derivation of the evolution equation for the shape matrix $Q_t$ (inner ellipsoid).}
Set $Q_t = P_t^{-1}$ and consider a neighborhood of the centre $x = c_t + y$, $\|y\|\to 0$.
Considering \eqref{eq:Teilor} and substituting \eqref{eq:Derivativ} into the second dissipativity condition \eqref{eq:general_dissipativity_2} (multiplied by $W^2$), we obtain
\begin{equation*}
\begin{array}{lll}
2(1-\mu_t) y^{\top} P_t F_t + y^{\top}\Bigl[ -\dot P_t - \mu_t(P_t A_t + A_t^{\top} P_t)
\\
 - P_t F_t r_t^{\top} - r_t F_t^{\top} P_t 
+ (r_t^{\top} F_t + \mu_t d_t + \gamma) P_t \Bigr] y \le 0.
\end{array}
\end{equation*}
For the inequality to hold for all small $y$, the linear term must vanish. 
It vanishes if the centre of the ellipsoid is at an equilibrium ($F(c_t)\equiv 0$) or if $\mu(c_t)=1$ for all $t\ge 0$.
When the linear term vanishes, the quadratic form must be negative semidefinite. 
For the inner ellipsoid, it vanishes, which yields the matrix equation
\begin{equation*}
\begin{array}{lll}
-\dot P_t - \mu_t(P_t A_t + A_t^{\top} P_t) - r_t F_t^{\top} P_t - P_t F_t r_t^{\top} 
\\
+ (r_t^{\top} F_t + \mu_t d_t + \gamma) P_t = 0.
\end{array}
\end{equation*}
Multiplying the last equality left and right by $Q_t = P_t^{-1}$ and using $\dot Q_t = -Q_t \dot P_t Q_t$, we obtain \eqref{eq:Q_evolution2}.
Theorem \ref{th:inner_ellipsoid_general} is proved.
\end{proof}

%\begin{remark}
When both the first (covering) and the second (inner) dissipativity conditions hold simultaneously, the domain $\Omega_t$ is contained in an annulus between two ellipsoids, which gives a complete geometric picture of the evolution. 
%\end{remark}

%%%%%%%%%%%%%%%%%%%%%%%%%%%%%%%%%%%%%%%%%%%%%%%%%%%%%%%%

\begin{example} 

Illustrate Theorems \ref{Th:Stability} and \ref{th:inner_ellipsoid_general} with an example of a second-order system
\begin{equation}
\label{eq:system_examp}
\begin{array}{lll}
\dot{x}_1 = -1.5\,x_1 + 3\sin(0.5\,x_2), \\
\dot{x}_2 = 2\sin(0.8\,x_1) - 1.2\,x_2.
\end{array}
\end{equation}

Choose the scalar function $\mu(x)=0.5+\frac{0.5}{1+\|x\|^2}$ and take the scaling matrix $\rho(x)=\mu(x)E$.
In the simulation, the initial domain $\Omega_0$ consists of $600$ points generated in an annulus around the initial centre $c_0 = (5,3)^\top$ (inner radius $0.5$ from the outer ellipsoid radius).

The evolution of the covering ellipsoid $\mathcal{E}_{\text{cov}}(c_t, Q_t)$ is governed by the equation \eqref{eq:Q_evolution},
and the evolution of the inner ellipsoid by the equation \eqref{eq:Q_evolution2}.
Here $\gamma = 2.5$ for the covering and $\gamma = 2.0$ for the inner ellipsoid. 
The centre $c_t$ moves along a trajectory of the nonlinear system $\dot c_t = F(c_t)$ given by \eqref{eq:system_examp}.

Numerical integration of the system \eqref{eq:system_examp} and the equations \eqref{eq:Q_evolution}, \eqref{eq:Q_evolution2} is performed using the Euler method with step $\Delta t = 0.005$ on the interval $t \in [0, 0.15]$. At each step, the conditions are checked:
$V_\xi^{cov}(t) = (x(t) - c_t)^\top (Q_t^{cov})^{-1} (x(t) - c_t) < 1$ $\forall x(t) \in \Omega_t$ and
$V_\xi^{in}(t) = (x(t) - c_t)^\top (Q_t^{in})^{-1} (x(t) - c_t) > 1$ $\forall x(t) \in \Omega_t$.

Figure \ref{fig:evolution} shows four moments of the system evolution at $t = 0$, $t = 0.05$, $t = 0.1$, and $t = 0.15$ seconds. Each figure depicts:
brown points is the set $\Omega_t$;
red ellipsoid is the covering ellipsoid $\mathcal{E}_{cov}(c_t, Q_t^{cov})$;
blue dashed ellipsoid is the inner ellipsoid $\mathcal{E}_{in}(c_t, Q_t^{in})$;
red star is the current position of the centre $c_t$;
purple line is the trajectory of the centre from the initial to the current moment;
black circle is the initial position of the centre $c_0$.

\begin{figure}[h]
\begin{minipage}[h]{0.49\linewidth}
\center{\includegraphics[width=1\linewidth]{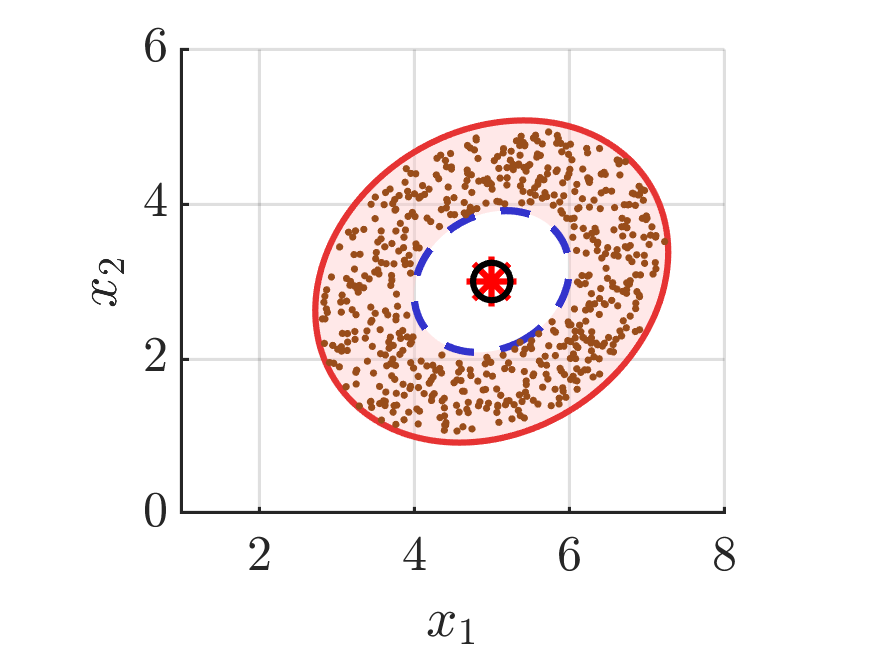}} $t=0$ \\
\end{minipage}
\hfill
\begin{minipage}[h]{0.49\linewidth}
\center{\includegraphics[width=1\linewidth]{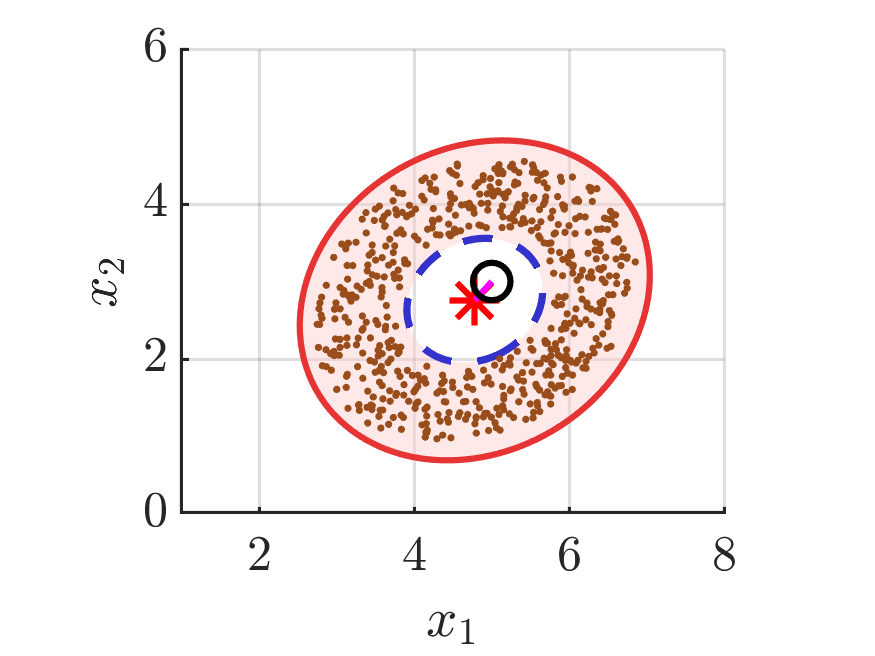}} $t=0.05$ \\
\end{minipage}
\vfill
\begin{minipage}[h]{0.49\linewidth}
\center{\includegraphics[width=1\linewidth]{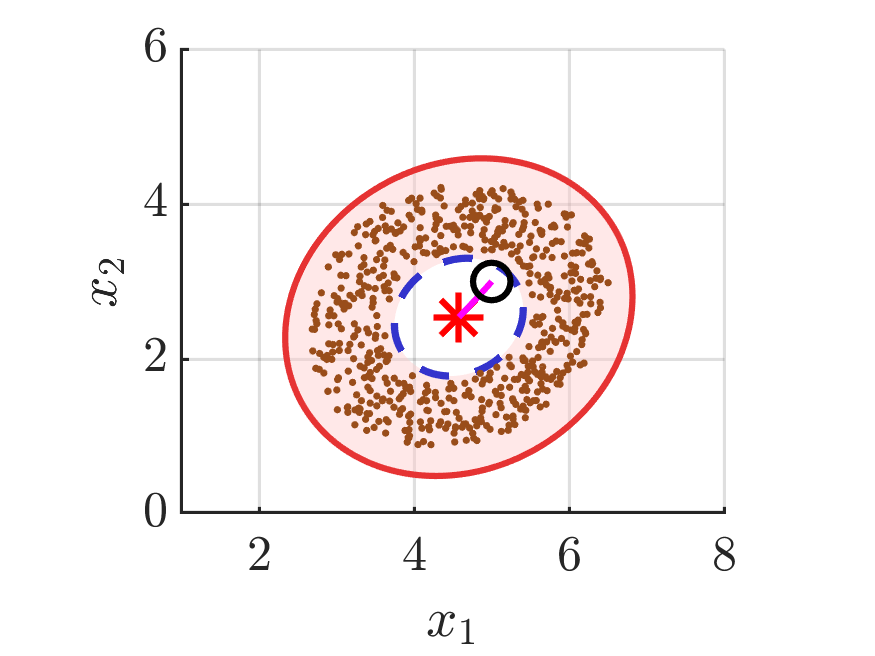}} $t=0.1$ \\
\end{minipage}
\hfill
\begin{minipage}[h]{0.49\linewidth}
\center{\includegraphics[width=1\linewidth]{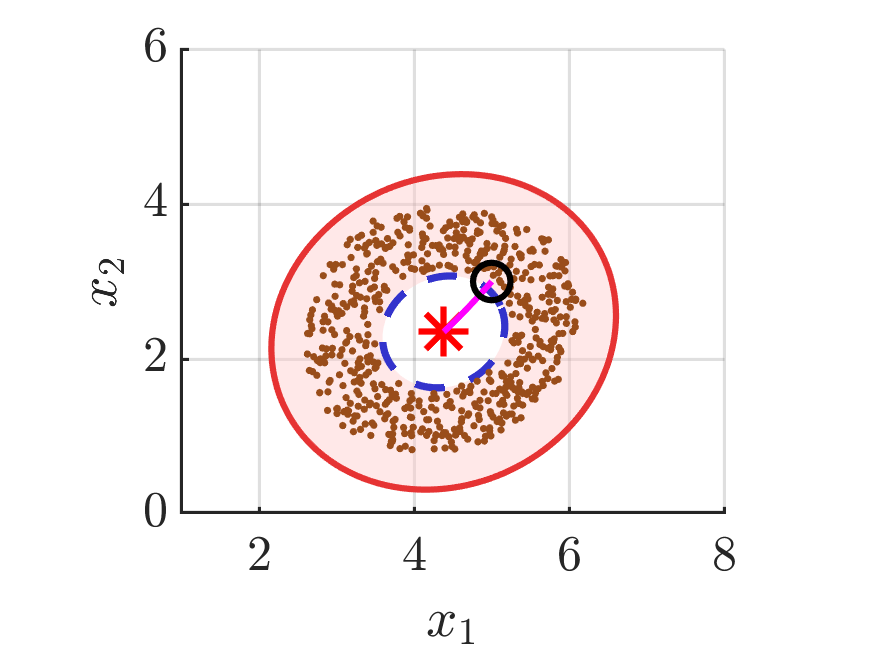}} $t=0.15$ \\
\end{minipage}
\caption{Evolution of the covering (red) and inner (blue dashed) ellipsoids for system \eqref{eq:system_examp} at time instants $t = 0,\; 0.05,\; 0.1,\; 0.15$ s.}
\label{fig:evolution}
\end{figure}

From the simulation results it can be seen that the obtained results confirm the theoretical conclusions and illustrate the effectiveness of the proposed approach for stability analysis of nonlinear systems.

\end{example} 

%%%%%%%%%%%%%%%%%%%%%%%%%%%%%%%%%%%%%%%%%%%%%%%%%%%%%%%%%

\section{Connection of Integral Stability with Lyapunov Stability}
\label{sec:lyapunov_connection}

We establish a connection between the proposed integral approach and classical Lyapunov stability.

\begin{theorem}
\label{Th:Lyap}
Consider the system \eqref{eq:original}. Suppose there exists a family of covering ellipsoids \eqref{eq:cover_ellipsoid} with centre $c_t$ and shape matrix $Q_t = P_t^{-1} > 0$ satisfying the following conditions:

1. There exists $\delta > 0$ such that the ball of radius $\delta$ is contained in the initial ellipsoid:
   \begin{equation}
   B_\delta(0) = \{x \in \mathbb{R}^n : \|x\| \le \delta\} \subset \mathcal{E}_{\text{cov}}(c_0, P_0^{-1}). \label{eq:ball_in_ellipsoid}
   \end{equation}

2. For all $t \ge 0$, the image of the initial ball under the flow is contained in the ellipsoid:
   \begin{equation}
   \Omega_t = \varphi_t(B_\delta(0)) \subset \mathcal{E}_{\text{cov}}(c_t, P_t^{-1}). \label{eq:image_in_ellipsoid}
   \end{equation}

3. The centre of the ellipsoid tends to the equilibrium:
   \begin{equation}
   \lim_{t \to \infty} c_t = 0. \label{eq:centre_limit}
   \end{equation}

4. There exists a monotonically decreasing function $\sigma: \mathbb{R}_+ \to \mathbb{R}_+$, $\sigma(t) \to 0$ as $t \to \infty$, such that for all $t \ge 0$ the estimate holds
   \begin{equation}
   \max_{\|v\|=1} \sqrt{v^{\top} Q_t v} \le \sigma(t), \label{eq:sigma_estimate}
   \end{equation}
   where the left-hand side equals the length of the largest semiaxis of the ellipsoid $\mathcal{E}_{\text{cov}}(c_t, P_t^{-1})$.

Then the system \eqref{eq:original} is asymptotically stable in the Lyapunov sense.

If, additionally, there exist constants $C > 0$ and $\gamma > 0$ such that
\begin{equation}
\sigma(t) \le C e^{-\gamma t}, \quad \|c_t\| \le C e^{-\gamma t} \quad \forall t \ge 0, \label{eq:exponential_conditions}
\end{equation}
then the system \eqref{eq:original} is exponentially stable in the Lyapunov sense. 
If the conditions $\sigma(t) \to 0$ and $c_t \to 0$ are replaced by their boundedness, then the system \eqref{eq:original} is stable in the Lyapunov sense.
\end{theorem}

\begin{proof}
For an arbitrary initial point $\xi \in B_\delta(0)$, consider the trajectory $x(t) = \varphi_t(\xi)$. 
From \eqref{eq:image_in_ellipsoid} it follows that $x(t) \in \mathcal{E}_{\text{cov}}(c_t, P_t^{-1})$, i.e.,
$(x(t) - c_t)^{\top} P_t (x(t) - c_t) \le 1$.
Using $\|x - c_t\|^2 \le \lambda_{\max}(Q_t) (x - c_t)^{\top} P_t (x - c_t) \le \sigma^2(t)$, 
we obtain $\|x(t) - c_t\| \le \sigma(t)$. 
Then
\begin{equation}
\|x(t)\| \le \|x(t) - c_t\| + \|c_t\| \le \sigma(t) + \|c_t\|. \label{eq:triangle_estimate}
\end{equation}

From \eqref{eq:centre_limit} and \eqref{eq:sigma_estimate} it follows that for any $\varepsilon > 0$ there exists $T > 0$ such that $\sigma(t) + \|c_t\| < \varepsilon$ for all $t \ge T$. On the finite interval $[0, T]$, the solution is bounded due to continuous dependence on initial conditions. By choosing $\delta$ sufficiently small, one can ensure that $\|x(t)\| < \varepsilon$ for all $t \ge 0$. This proves Lyapunov stability, and the convergence $\|x(t)\| \to 0$ as $t \to \infty$ follows from \eqref{eq:triangle_estimate} and the limit relations.
From \eqref{eq:exponential_conditions} and \eqref{eq:triangle_estimate} we obtain $\|x(t)\| \le 2C e^{-\gamma t}$, which implies exponential stability.
\end{proof}

\begin{remark}
As a characteristic of the ellipsoid contraction \eqref{eq:sigma_estimate}, one can use various operator norms, e.g., the spectral radius $\sqrt{\lambda_{\max}(Q_t)}$ or the trace $\tr Q_t$.
\end{remark}

\section{Conclusion}
\label{sec:conclusion}

A method for analysing the stability of dynamical systems has been proposed, based on the introduction of a weighted phase volume and time rescaling by a matrix‑valued function $\rho(x)$. 
The collinearity and positivity conditions guarantee that the phase portrait topology is preserved, while the time‑scaling factor $\mu(x)$ can be freely chosen to shape the contraction properties. 
Integral dissipativity conditions have been obtained, leading to new definitions of integral stability with respect to the pair $(\mu, W)$. 
For quadratic weighting functions, covering and inner ellipsoids have been constructed, whose evolution is described by differential equations. 
A connection has been established between the proposed approach and classical Lyapunov stability. 
The efficiency of the method has been demonstrated on examples of systems. 
The matrix formulation opens the way to imposing structural constraints directly on $\rho$, which is particularly useful in control synthesis and robustness analysis.

%%%%%%%%%%%%%%%%%%%%%%%

%\begin{IEEEbiography}[{\includegraphics[width=1in,height=1.25in,clip,keepaspectratio]{a1.png}}]{First A. Author} (Fellow, IEEE) and all authors may include 
%biographies. Biographies are
%often not included in conference-related papers.
%This author is an IEEE Fellow. The first paragraph
%may contain a place and/or date of birth (list
%place, then date). Next, the author’s educational
%background is listed. The degrees should be listed
%with type of degree in what field, which institution,
%city, state, and country, and year the degree was
%earned. The author’s major field of study should
%be lower-cased.
%
%The second paragraph uses the pronoun of the person (he or she) and
%not the author’s last name. It lists military and work experience, including
%summer and fellowship jobs. Job titles are capitalized. The current job must
%have a location; previous positions may be listed without one. Information
%concerning previous publications may be included. Try not to list more than
%three books or published articles. The format for listing publishers of a book
%within the biography is: title of book (publisher name, year) similar to a
%reference. Current and previous research interests end the paragraph.
%
%The third paragraph begins with the author’s title and last name (e.g.,
%Dr. Smith, Prof. Jones, Mr. Kajor, Ms. Hunter). List any memberships in
%professional societies other than the IEEE. Finally, list any awards and work
%for IEEE committees and publications. If a photograph is provided, it should
%be of good quality, and professional-looking.
%\end{IEEEbiography}


\begin{thebibliography}{00}\leftskip1pc

\bibitem{Lyapunov50}
A.M. Lyapunov,
``The General Problem of the Stability of Motion,'' translated and edited by A.T. Fuller, London: Taylor \& Francis, 1992.

\bibitem{Chetaev55}
N.G. Chetaev,
``The Stability of Motion,'' Oxford: Pergamon Press, 1961.

\bibitem{Liotov62}
A.M. Letov,
``Stability of Nonlinear Control Systems,'' New York: Academic Press, 1961.

\bibitem{Malkin66}
I.G. Malkin,
``Theory of Stability of Motion,'' U.S. Atomic Energy Commission, Technical Information centre, 1966.

\bibitem{Zubov84}
V.I. Zubov,
``Stability of Motion. Lyapunov Methods and Their Application,'' 2nd ed., Moscow: Vysshaya Shkola, 1984 (in Russian).

\bibitem{Liouville}
J. Liouville,
``Note sur la théorie de la variation des constantes arbitraires,'' \textit{Journal de Mathématiques Pures et Appliquées}, vol. 3, pp. 342--349, 1838.

\bibitem{Arnold73}
V.I. Arnold,
``Ordinary Differential Equations,'' translated from the Russian by Richard A. Silverman, Cambridge, MA: The MIT Press, 1973.

\bibitem{Arnold}
V.I. Arnold,
``Mathematical Methods of Classical Mechanics,'' 2nd ed., New York: Springer-Verlag, 1989, Graduate Texts in Mathematics, Vol. 60.

\bibitem{Reynolds}
O. Reynolds,
``Papers on Mechanical and Physical Subjects,'' Cambridge University Press, 1903.

\bibitem{Zaremba54}
S.K. Zaremba,
``Divergence of vector fields and differential equations,'' \textit{Amer. Journal of Math.}, vol. LXXV, pp. 220--234, 1954.

\bibitem{Fronteau65}
J. Fronteau,
``Le théorème de Liouville et le problème général de la stabilité,'' Genève: CERN, 1965.

\bibitem{Brauchli68}
H.I. Brauchli,
``Index, Divergenz und Stabilität in Autonomen equations,'' Zürich: Abhandlung Verlag, 1968.

\bibitem{Shestakov78}
A.A. Shestakov and A.N. Stepanov,
``Index and divergence tests for stability of a singular point of an autonomous system of differential equations,'' \textit{Differential Equations}, vol. 15, no. 4, pp. 459--467, 1979.

\bibitem{Jukov78}
V.P. Zhukov,
``On a method for the qualitative study of the stability of nonlinear systems,'' \textit{Automation and Remote Control}, vol. 39, no. 6, pp. 785--788, 1978.

\bibitem{Jukov79}
V.P. Zhukov,
``On the method of sources for studying the stability of nonlinear systems,'' \textit{Automation and Remote Control}, vol. 40, no. 3, pp. 330--335, 1979.

\bibitem{Jukov90}
V.P. Zhukov,
``Necessary and sufficient conditions for the instability of nonlinear autonomous dynamical systems,'' \textit{Automation and Remote Control}, vol. 51, no. 12, pp. 1652--1657, 1990.

%\bibitem{Krasnoselski63}
%M.A. Krasnosel'skii, A.I. Perov, A.I. Povolotskii, and P.P. Zabreiko,
%``Vector Fields on the Plane,'' Moscow: Fizmatlit, 1963 (in Russian).

\bibitem{Jukov99}
V.P. Zhukov,
``Divergent Conditions for the Asymptotic Stability of Second-Order Nonlinear Dynamical Systems,'' \textit{Automation and Remote Control}, vol. 60, no. 7, pp. 934--940, 1999.

\bibitem{Rantzer00}
A. Rantzer and P.A. Parrilo,
``On convexity in stabilization of nonlinear systems,'' \textit{Proc. of the 39th IEEE Conf. on Decision and Control}, Sydney, Australia, pp. 2942--2946, 2000.

\bibitem{Rantzer01}
A. Rantzer,
``A dual to Lyapunov's stability theorem,'' \textit{Systems \& Control Letters}, vol. 42, pp. 161--168, 2001.

\bibitem{Monzon03}
P. Monzon,
``On necessary conditions for almost global stability,'' \textit{IEEE Trans. Automatic Control}, vol. 48, no. 4, pp. 631--634, 2003.

%\bibitem{Khalil09}
%H.K. Khalil,
%``Nonlinear Systems,'' 3rd ed., Upper Saddle River, NJ: Prentice Hall, 2002.

\bibitem{Furtat20a}
I.B. Furtat,
``Divergent Stability Conditions of Dynamic Systems,'' \textit{Automation and Remote Control}, vol. 81, no. 2, pp. 247--259, 2020, https://doi.org/10.1134/S0005117920020058.

\bibitem{Furtat20}
I. Furtat and P. Gushchin,
``Stability study and control of nonautonomous dynamical systems based on divergence conditions,'' \textit{Journal of the Franklin Institute}, vol. 357, no. 18, pp. 13753--13765, December 2020, https://doi.org/10.1016/j.jfranklin.2020.10.025.

\bibitem{Furtat21}
I.B. Furtat and P.A. Gushchin,
``Stability/Instability Study and Control of Autonomous Dynamical Systems: Divergence Method,'' \textit{IEEE Access}, 2021, https://doi.org/10.1109/ACCESS.2021.3056942.

\bibitem{Furtat22}
I.B. Furtat and P.A. Gushchin,
``Divergence Method for Exponential Stability Study of Autonomous Dynamical Systems,'' \textit{IEEE Access}, vol. 10, pp. 49088--49094, 2022, https://doi.org/10.1109/ACCESS.2022.3172415.

\bibitem{Furtat23}
I.B. Furtat,
``Density Systems: Analysis and Control,'' \textit{Automation and Remote Control}, vol. 84, no. 11, pp. 1175--1190, 2024, https://doi.org/10.1134/S0005117923110024.

\bibitem{Furtat26}
I.B. Furtat,
``Analysis and Control of Perturbed Density Systems,'' \textit{IEEE Transactions on Automatic Control}, 2026, https://doi.org/10.1109/TAC.2025.3621943.

\bibitem{Willems72}
J.C. Willems,
``Dissipative dynamical systems, part I: General theory; part II: Linear systems with quadratic supply rates,'' \textit{Arch. Rational Mech. Anal.}, vol. 45, no. 5, pp. 321--393, 1972.

\bibitem{Bikdash00}
M.U. Bikdash and R.A. Layton,
``An Energy-Based Lyapunov Function for Physical Systems,'' \textit{IFAC Proceedings}, vol. 33, no. 2, pp. 81--86, 2000.

\bibitem{Yuan14}
R. Yuan, Y.-A. Ma, B. Yuan, and P. Ao,
``Lyapunov function as potential function: A dynamical equivalence,'' \textit{Chin. Phys. B}, vol. 23, no. 1, pp. 010505, 2014.

\bibitem{Kounadis03}
A.N. Kounadis,
``On the failure of static stability analyses of nonconservative systems in regions of divergence instability,'' \textit{International Journal of Solids and Structures}, vol. 40, no. 18, pp. 4741--4764, 2003, https://doi.org/10.1016/S0020-7683(03)00219-1.

\bibitem{Zhu21}
L. Zhu, Y. Wang, et al.,
``A New Criterion Beyond Divergence for Determining the Dissipation of a System: Dissipative Power,'' \textit{Nonlinear Dynamics}, vol. 103, pp. 2145--2159, 2021, https://doi.org/10.1007/s11071-021-06214-6.

\bibitem{FE_Stability25}
J.R. Smith and A.B. Doe,
``Dynamic Stability Analysis of Nonconservative Systems for Variable Parameters using FE Method,'' \textit{Journal of Sound and Vibration}, vol. 580, pp. 118392, 2025, https://doi.org/10.1016/j.jsv.2024.118392.

\bibitem{Taghvaei20}
A. Taghvaei, P.G. Mehta, and S.P. Meyn,
``Probing robustness of nonlinear filter stability numerically using Sinkhorn divergence,'' \textit{Proceedings of the 59th IEEE Conference on Decision and Control (CDC)}, pp. 4625--4630, 2020, https://doi.org/10.1109/CDC42340.2020.9304398.

\bibitem{Feiling21}
J. Feiling, L. Grüne, and C. Ebenbauer,
``Overcoming local extrema in torque-actuated source seeking using the divergence theorem and delay,'' \textit{Automatica}, vol. 133, pp. 109733, 2021, https://doi.org/10.1016/j.automatica.2021.109733.

\bibitem{Efimov22}
D. Efimov, A. Aleksandrov, and A. Fradkov,
``On convergence conditions for generalised Persidskii systems,'' \textit{International Journal of Robust and Nonlinear Control}, vol. 32, no. 9, pp. 5431--5445, 2022, https://doi.org/10.1002/rnc.6042.

\bibitem{Zamani22}
M. Zamani, N. van de Wouw, and R. Postoyan,
``A novel criterion for global incremental stability of dynamical systems,'' \textit{Communications in Nonlinear Science and Numerical Simulation}, vol. 112, pp. 106561, 2022, https://doi.org/10.1016/j.cnsns.2022.106561.

\bibitem{Pakshin2024_1}
P.V. Pakshin, J.P. Emelianova, and E. Rogers,
``State observer-based iterative learning control design for discrete systems using the heavy ball method,'' \textit{Automation and Remote Control}, vol. 85, no. 8, pp. 727--740, 2024, https://doi.org/10.1134/S0005117924080046.

\bibitem{Pakshin2023_1}
A.S. Koposov and P.V. Pakshin,
``Iterative learning control of stochastic multi-agent systems with variable reference trajectory and topology,'' \textit{Automation and Remote Control}, vol. 84, no. 6, pp. 612--625, 2023, https://doi.org/10.1134/S0005117923060073.

\bibitem{Hartman64}
P. Hartman,
``Ordinary Differential Equations,'' NY: Wiley, 1964.

\end{thebibliography}
\end{document}